\documentclass[12pt,reqno]{amsart}
\usepackage{amssymb}
\usepackage{amsmath}
\usepackage{amsmath}
\usepackage{amsthm}
\usepackage{amsfonts}
\usepackage{graphicx}
\usepackage{enumerate}
\usepackage{bm}
\usepackage{bbm} 
\usepackage{graphicx}
\usepackage{dsfont}
\usepackage[authoryear,longnamesfirst]{natbib}
\usepackage{fullpage}
\usepackage{subcaption}
 \usepackage{xr}
\usepackage{multirow}

\usepackage{multicol}
\usepackage{tikz}

\numberwithin{equation}{section}
\newtheorem{theorem}{Theorem}

\newtheorem{corollary}{Corollary}

\newtheorem{definition}{Definition}

\newtheorem{lemma}{Lemma}
\theoremstyle{definition}
\newtheorem{assumption}{Assumption}

\newtheorem{remark}{Remark}

\newcommand{\supp}{\mathrm{supp}}
\renewcommand{\tilde}{\widetilde}

\newcommand{\Var}[0]{\mathrm{Var}}

\newcommand{\calC}[0]{\mathcal{C}}
\newcommand{\calF}[0]{\mathcal{F}}
\newcommand{\calG}[0]{\mathcal{G}}
\newcommand{\calH}[0]{\mathcal{H}}

\newcommand{\calS}[0]{\mathcal{S}}

\newcommand{\calX}[0]{\mathcal{X}}
\newcommand{\calY}[0]{\mathcal{Y}}

\newcommand{\E}[0]{\mathbb{E}}

\newcommand{\R}[0]{\mathbb{R}}

\newcommand{\N}{\mathbb{N}}

\newcommand{\sumi}{\sum_{i=1}^n }

\newcommand{\Real}{\mathbbm R}

\newcommand{\Op}{O_P}

\newcommand{\op}{o_P}

\renewcommand{\hat}{\widehat}
\newcommand{\Ink}{I_{n}}
\newcommand{\ijI}{{(i,j)\in I_{n}}}
\newcommand{\Xij}{{X_{ij}}}

\newcommand{\sg}{{\text{sg}}}

\allowdisplaybreaks
\begin{document}

\title{Empirical likelihood and uniform convergence rates for dyadic kernel density estimation}
\thanks{We thank Co-Editor Jianqing Fan, the Associate Editor, and three anonymous referees for their constructive comments that have helped us significantly improved the paper. We are indebted to Matias Cattaneo, the anonymous AE, and two anonymous referees for pointing out an error in the earlier version and for suggesting a fix. We benefited from useful comments and discussions with Bruce Hansen, Yuya Sasaki, and the participants at Bristol Econometrics Study Group. We thank Yukitoshi Matsushita for sharing their MATLAB simulation code. 
All remaining errors are ours.}
\author[H. D. Chiang]{Harold D. Chiang}
\author[B. Y. Tan]{Bing Yang Tan}

\date{First arXiv version: 16 Oct 2020. This version: \today}

\address[H. D. Chiang]{
Department of Economics, University of Wisconsin-Madison\\
 William H. Sewell Social Science Building, 1180 Observatory Drive,
Madison, WI 53706, USA.}
\email{hdchiang@wisc.edu}

\address[B. Y. Tan]{
Global Asia Institute, National University of Singapore \\
Block S17, Level 3, \#03-01 10 Lower Kent Ridge Road \\
Singapore 11907}
\email{tanby91@nus.edu.sg}

\begin{abstract}
This paper studies the asymptotic properties of and alternative inference methods for kernel density estimation (KDE) for dyadic data. We first establish uniform convergence rates for dyadic KDE. Secondly,  we propose a modified jackknife empirical likelihood procedure for inference. The proposed test statistic is asymptotically pivotal regardless of presence of dyadic clustering. The results are further extended to cover the practically relevant case of incomplete dyadic data.
Simulations show that this modified jackknife empirical likelihood-based inference procedure delivers precise coverage probabilities even with modest sample sizes and with incomplete dyadic data. Finally, we illustrate the method by studying airport congestion in the United States.

\medskip


\medskip
\end{abstract}

\maketitle

\newpage
\section{Introduction}
Consider a weighted undirected random graph consisting of vertices $i\in\{1,...,n\}$ and edges consisting of real-valued random variables $(X_{ij})_{i,j\in\{1,...n\}, i< j}$. Suppose $X_{ij}$ are identically distributed with density $f(\cdot)=f_{X_{12}}(\cdot)$. The dyadic kernel density estimator (KDE) of \cite{GrahamNiuPowell2019}, evaluated at a design point $x\in \supp(X_{12})$, is defined as
\begin{align*}
\hat f_{n}(x)= {n\choose 2}^{-1}\sum_{1\le i< j \le n}\frac{1}{h_n}K\left(\frac{x-X_{ij}}{h_n}\right).
\end{align*}
\par
Applications of dyadic data are numerous in different fields of the sciences and social sciences, such as bilateral trade, animal migrations, refugee diasporas, transportation networks between different cities, friendship networks between individuals, intermediate product sales between firms, research and development partnerships across different organizations, electricity grids across different states, and social networks between legislators, to list a few. See, e.g. \cite{Graham2019Handbook} for a recent and comprehensive review. This paper investigates the asymptotic properties of and proposes new, improved inference methods for dyadic KDE with both complete and incomplete data.\\
\par
The main contributions of this paper are two-fold. First, we establish uniform convergence rates for dyadic KDE under general conditions. These rates are new; no uniform rate is previously  known for KDE under dyadic sampling in the literature. Second, we develop inference methods by adapting a modified version of the jackknife empirical likelihood (JEL) proposed by \cite{MatsushitaOtsu2020}. The proposed modified JEL (mJEL) procedure is shown to be asymptotically valid regardless of the presence of dyadic clustering. In practice, dyadic clustered data are often incomplete or contain missing values. We extend our modified JEL procedure to cover the practically relevant case of incomplete data under the missing at random assumption.
Extensive simulation studies show that this modified JEL inference procedure does not suffer from the unreliable finite sample performance of the analytic variance estimator, delivers precise coverage probabilities even under modest sample sizes, and is robust against both dyadic clustering and i.i.d. sampling for both complete and incomplete dyadic data.  Finally, we illustrate the method by studying the distribution of congestion across different flight routes using data from the United States Department of Transportation (US DOT).\\
\par
In an important recent work, \cite{GrahamNiuPowell2019} first propose and examine a
nonparametric density estimator for dyadic random variables.
Focusing on pointwise asymptotic behaviors, they demonstrate that asymptotic normality at a fixed design point can be established at a parametric rate with respect to number of vertices, an unique feature of dyadic KDE. This implies that its asymptotics are robust to a range of bandwidth rates. They also point out that degeneracy that leads to a non-Gaussian limit, such as the situation pointed out in \cite{Menzel2020}, does not occur for dyadic KDE under some mild bandwidth conditions. For inference, they propose an analytic variance estimator that corresponds to the one proposed by \cite{FafchampsGubert2007} (henceforth FG) under parametric models and establish its asymptotic validity. { 
In their simulation, they discover that the FG variance estimator often leads to imprecise coverage with moderate sample sizes. 
Our proposed modified JEL procedure for dyadic KDE complements the findings of \cite{GrahamNiuPowell2019} by providing an alternative that has improved finite sample performance.\\
\par

}
 In practice, the presence of missing edges is an important feature of network datasets, and a researcher working with network data often observes incomplete sets of edges. To our knowledge, none of the existing theoretical work in the literature that focuses on the asymptotics of dyadic KDE has tackled this issue. Under the missing at random assumption, we extend our theory for modified JEL to cover this practically relevant situation. 
\par
Ever since the seminal papers of \cite{Owen1988,Owen1990}, empirical likelihood has been extensively studied in the literature.  For a textbook treatment of classical theory, see \cite{Owen200Book}.  Some recent developments under conventional asymptotics include \cite{HjortMcKeagueVan_Keilegom2009AoS} and \cite{BravoEscancianoVan_Keilegom2020AoS}, and many more. See \cite{chen2009review} for a recent review. Jackknife empirical likelihood was first proposed in the seminal work of
 \cite{JingYuanZhang2009} for parametric $U$-statistics and has since been studied by \cite{GongPengQi2010JMA}, \cite{PengOiVan_Keilegom2012}, \cite{Peng2012CanJS}, \cite{WangPengQi2013Sinica}, \cite{ZhangZhao2013JMA}, \cite{MatsushitaOtsu2018JER} and \cite{ChenTabri2020}, to list a few, for different applications.  Most of the aforementioned works focus on actual $U$-statistics under i.i.d. sampling.  \cite{MatsushitaOtsu2020} demonstrate that JEL can not only be modified for certain unconventional asymptotics of i.i.d. observations, but can also be applied to study edge probabilities for both sparse and dense dyadic networks studied in \cite{BickelChenLevina2011}. This is achieved by incorporating the bias-correction idea of \cite{Hinkley1978}, \cite{EfronStein1981}. In a context where controlling bias and spikes - from the presence of the bandwidth in kernel estimation - presents extra challenges, our modified JEL for dyadic KDE provides a nonparametric counterpart to the network asymptotic results in \cite{MatsushitaOtsu2020}. \\
\par
{ 
Theoretically speaking, the modified JEL is closely related to the growing literature that utilizes
the idea of accounting for the contributions from both linear and quadratic terms in the Hoeffding-type decomposition of the statistics. 
This idea emerges in the literature that studies
 ``small-bandwidth asymptotics," in the context of density-weighted average derivative by \cite{,CattaneoCrumpJansson2014bootstrap,CattaneoCrumpJansson2014ET} and \cite{CattaneoCrumpJansson2013JASA}, for the sake of robustness over a wider  range of bandwidths. This idea has since been explored in various other contexts such as an inference problem in partially linear models with many regressors in \cite{cattaneo2018alternative,cattaneo2018inference}, as well as the various other examples explored in \cite{MatsushitaOtsu2020}.}\\ 
 \par
Our results also complement the uniform convergence rates results for kernel type estimators under different dependence settings studied in \cite{EinmahlMason2000}, \cite{EinmahlMason2005}, \cite{DonyMason2008}, \cite{Hansen2008ET}, \cite{Kristensen2009ET} and so forth. Theoretically, the $\sqrt{n}$-uniform rate of dyadic KDE share common underlying features with the kernel density estimator of \cite{EscancianoJacho-Chavez2012}, as well as the residual distribution estimator proposed in \cite{AkritasVan_Keilegom2001}, both under i.i.d. sampling.\\   

\par
{ 
In a seminal work, \cite{FafchampsGubert2007} proposed a dyadic robust estimator for regression models, which has since become the benchmark for inference under dyadic clustering. Related works include \cite{FrankSnijders1994} and \cite{SnijdersBorgatti1999}, \cite{AronowSamiiAssenova2015} \cite{CameronMiller2014} in different contexts.
Asymptotics for statistics of dyadic random arrays are studied by \cite{DDG2020} under a general empirical processes setting and by \cite{ChiangKatoSasaki2020} in high-dimensions. The former show the validity of a modified pigeonhole bootstrap (cf. \cite{Mccullagh2000} and \cite{Owen2007}) while the latter utilize a multiplier bootstrap, both under a non-degeneracy condition. On the other hand, for two-way separately exchangeable arrays and linear test-statistics, \cite{Menzel2020} proposed a conservative bootstrap that is valid uniformly even under potentially non-Gaussian degeneracy scenarios. 
In this paper, we only need to focus on the non-degenerate and Gaussian degenerate cases since the non-Gaussian degeneracy scenario is not a concern for dyadic KDE; \cite{GrahamNiuPowell2019} pointed out that non-Gaussian degeneracy can be ruled out under mild bandwidth conditions.\\
\par
More recently, \cite{cattaneo2022} study uniform convergence rate and uniform inference for dyadic KDE. Their uniform convergence rate results further improved upon the result by allowing uniformity over the whole support for compactly supported
data and enabling boundary-adaptive estimation. They further derive the minimax rate of uniform convergence for density estimation with dyadic
	data and show that the dyadic KDE is, under appropriate conditions, minimax-optimal. For  inference, they obtain uniform distribution theory and provide bias-corrected $t$-statistics-based uniform confidence bands. Their methods differ from the ones we propose, which are based on jackknife empirical likelihood.}

\subsection{Notation and Organization}
Let $|A|$ be the cardinality of a finite set $A$. Denote the uniform distribution on $[0,1]$ as $U[0,1]$.  For $g:\calS\to \R$ and $\calX\subseteq \calS$, denote $\|g\|_{\calX}=\sup_{x\in\calX}|g(x)|$.  Write $a_n\lesssim b_n$ if there exists a constant $C>0$ independent of $n$ such that $a_n\le Cb_n$. Throughout the paper, the asymptotics should be understood as taking $n\to \infty$.\\
\par
The rest of this paper is organized as follows. In Section \ref{sec:model_estimator}, we introduce our model and the dyadic KDE estimator. Our main theoretical results of uniform convergence rates and the validity of different JEL statistics for both complete and incomplete dyadic random arrays are then introduced in Section \ref{sec:theory}. Section \ref{sec:simulation} contains the simulation studies while the empirical application can be found in Section \ref{sec:application}. A practical guideline on bandwidth choice and proofs of all the theoretical results are contained in the Supplementary Appendix.


\section{Model and estimator}\label{sec:model_estimator}
Consider a weighted undirected random graph.
Let $I_{n} = \{ (i,j) : 1 \le i < j \le n\}$ and $I_{\infty} = \bigcup_{n=2}^{\infty} I_{n}$. The random graph consists of an array $( X_{ij})_{(i,j) \in I_{\infty}}$ of real-valued random variables that is generated by
\begin{align}
X_{ij} = \mathfrak f (U_i,U_j ,U_{\{i,j\}}),\quad  (U_i)_{i\in\N}, \:(U_{\{i,j\}})_{i,j\in \N}  \stackrel{i.i.d.}{\sim} U[0,1]\label{eq:Aldous-Hoover}
\end{align}
for some Borel measurable map $\mathfrak{f}:[0,1]^3\to \R$ that is symmetric in the first two arguments. Note that $U_i$ and $U_{\{\iota,\jmath\}}$ are independent for all $i\in\N$, $(\iota,\jmath)\in I_{\infty}$. While (\ref{eq:Aldous-Hoover}) may seem to be a specific structural assumption, it is implied by the following low-level condition of joint exchangeability via the Aldous-Hoover-Kallenberg representation, see e.g. \citet[Theorem 7.22]{Kallenberg2006}.
\begin{definition}[Joint exchangeability]\label{a:JE}
	An $(X_{ij})_{(i,j) \in I_{\infty}}$ is said to be jointly exchangeable if 
	for any permutation $\pi$ of $\N$, the arrays $(X_{ij})_{(i,j) \in I_{\infty}}$ and $(X_{\pi(i)\pi(j)})_{(i,j) \in I_{\infty}}$ are identically distributed.
\end{definition}

\begin{remark}[Identical distribution]
	Under this setting, $(X_{ij})_{(i,j)\in I_{\infty}}$ are identically distributed.
\end{remark}

Suppose that the researcher observes $(X_{ij})_{(i,j)\in I_{n}}$ for $n\ge 2$. The object of interest is the density function $f(\cdot)=f_{X_{12}}(\cdot)$ of $X_{12}$. We assume such a density function exists and is well-defined although some low-level sufficient conditions can be obtained. 
For a fixed design point $x\in\supp(X_{12})$ of interest, the dyadic KDE for $\theta=f(x)$ is defined as
\begin{align}
\hat \theta=\hat f_{n}(x)= {n\choose 2}^{-1}\sum_{i=1}^{n-1}\sum_{j= i+1}^n\frac{1}{h_n}K\left(\frac{x-X_{ij}}{h_n}\right)=\frac{1}{|\Ink|}\sum_{\ijI} K_{ij,n},\label{eq:KDE}
\end{align}
where $h_n>0$ is a $n$-dependent bandwidth, $K_{ij,n}:=h_n^{-1}K((x-X_{ij})/h_n)$ and $K(\cdot)$ is a kernel function. We will be more specific about the requirements for $h_n$ and $K$ in the following Section.

\section{Theoretical results}\label{sec:theory}
\subsection{Uniform convergence rates}
 { In the rest of the paper, denote $f_{X_{12}|U_1}$ for the conditional density function of $X_{12}$ given $U_1$ and $f_{X_{12}|U_1,U_2}$ for the conditional density function of $X_{12}$ given $U_1,U_2$. In addition, define the shorthand notations $f'(\cdot)=\partial f(\cdot)/\partial x$, $f''(\cdot)=\partial^2 f(\cdot)/\partial x^2$, $f_{X_{12}|U_1}'(\cdot|u)=\partial f_{X_{12}|U_1}(\cdot|u)/\partial x$, $f_{X_{12}|U_1}''(\cdot|u)=\partial^2 f_{X_{12}|U_1}(\cdot|u)/\partial x^2$. Similarly, $f_{X_{12}|U_1,U_2}'(\cdot|u_1,u_2)=\partial f_{X_{12}|U_1,U_2}(\cdot|u_1,u_2)/\partial x$, and $f_{X_{12}|U_1,U_2}''(\cdot|u_1,u_2)=\partial^2 f_{X_{12}|U_1,U_2}(\cdot|u_1,u_2)/\partial x^2$.} We first study the uniform convergence rates of dyadic KDE under the following assumptions.
\begin{assumption}[Uniform convergence rates]\label{a:rates}
	Suppose that { $\calX \subsetneq \supp(X_{12})$} is a compact interval, and suppose the following are satisfied:
\begin{enumerate}[(i)]
	\item We observe $(X_{ij})_{(i,j)\in I_n}$, a subset of $( X_{ij})_{(i,j) \in I_{\infty}}$ that is generated following (\ref{eq:Aldous-Hoover}) with $\Var(X_{12})>0$.
	\item In an open neighborhood of $\calX$, $f(\cdot)$ exists and is at least twice continuously differentiable and $\|f''\|_\calX =O(1)$, $f_{X_{12}|U_1}(\cdot|U)$ exists and is almost surely at least twice continuously differentiable, $\|f_{X_{12}|U_1}''\|_{\calX\times [0,1]} =O(1)$, $f_{X_{12}|U_1,U_2}(\cdot|U_1,U_2)$ exists and is almost surely continuously differentiable and $\|f_{X_{12}|U_1,U_2}'\|_{\calX\times [0,1]^2} =O(1)$.
	\item The kernel function $K$ is symmetric, right (or left) continuous, non-negative, of bounded variation, and satisfies $\|K\|_\infty=O(1)$, $\int K(u)du=1$, { $\int uK(u)du=0$}, and $0<\int u^2 K(u)du<\infty$.
	\item The sequence of bandwidths satisfies $h_n\to 0$ and { $nh_n\to \infty$}. 
\end{enumerate}
\end{assumption}

{  The support condition in Assumption \ref{a:rates} restricts the scope to be on a compact interval $\mathcal X$ that is strictly contained in the support of $X_{12}$. Assumption \ref{a:rates}(i) requires the observations to be generated following the dyadic structure discussed in Section \ref{sec:model_estimator}. Assumption \ref{a:rates}(ii) assumes existence of the conditional densities conditional on vertex-specific latent shocks and imposes standard smoothness conditions on the unknown density and conditional densities. Assumption \ref{a:rates}(iii) assumes a smooth second-order kernel. Finally, Assumption \ref{a:rates}(iv) requires the sequence of bandwidths to be converging to zero in an adequate range. 
}

\begin{theorem}[Uniform convergence rates for dyadic KDE]\label{thm:unif_rate}
	Suppose Assumption \ref{a:rates} holds, then with probability at least $1-o(1)$,
	\begin{align*}
	\sup_{x\in \calX}\left|\hat f_n(x)-f(x)\right|\lesssim h_n^2
	+\frac{1}{n^{1/2}}
	\end{align*}
	if  $\inf_{x\in\calX}\Var\left(f_{X_{12}|U_1}(x|U_1)\right)\ge\underline L>0$ for some $\underline L>0$ (called the non-degenerate case). Otherwise (called the degenerate case), with probability at least $1-o(1)$,
	\begin{align*}
	\sup_{x\in \calX}\left|\hat f_n(x)-f(x)\right|\lesssim h_n^2 +\frac{1}{n}
	+\sqrt{\frac{\log(1/h_n)}{n^2h_n}}.
	\end{align*}
\end{theorem}
A proof can be found in Section D.1 of the Supplementary Appendix.
\begin{remark}[Convergence rates]
	{ 	Theorem \ref{thm:unif_rate} shows the uniform convergence rates under both non-degenerate and degenerate scenarios. 
	In both cases, the bias is $O(h_n^2)$. In case of degeneracy, the uniform rate coincides with the independent sampling situation with sample size $2{n\choose 2}$. { Besides the bias term}, the component $\sqrt{\frac{\log(1/h_n)}{n^2h_n}}$ is the part in which the kernel plays a role (in adding the factor of $\frac{\log(1/h_n)}{h_n}$). In the non-degenerate case, the stochastic component is $O_p(n^{-1/2})$ because its leading term consists of a sample average of $\{\E[h_n^{-1 }K((x-X_{ij})/h_n)|U_i]\}_{i=1}^n$, and these elements are well-behaved under our assumptions.
	 These results agree with the observations made in \cite{GrahamNiuPowell2019} for the pointwise behaviors of the dyadic KDE. In a related paper, \cite{graham2021minimax} consider nonparametric regression with dyadic data  and establish uniform convergence rates for a Nadaraya-Watson type estimator under a different setting. More explicitly, their rates are nonparametric rates since the regressors considered are vertex-specific. 
 
}
\end{remark}

\begin{remark}[Recent improvement on the uniform convergence rates]\label{rem:recent}
 The recent paper of \cite{cattaneo2022} further improved upon the result by obtaining uniform convergence rates for dyadic KDE under a more general setting. Specifically, their uniformity is over the whole support for
compactly supported
data by carefully addressing the boundary issues, while Theorem \ref{thm:unif_rate} is established on a compact interval strictly contained in the support of $X_{12}$.
  They further show that dyadic KDE can be minimax optimal under some regularity conditions. 
\end{remark}


\subsection{Inference for complete dyadic data with JEL and modified JEL}
We now introduce JEL-based inference procedures for dyadic KDE.
Denote the leave-one-out and leave-two-out index sets $I_n^{(i)}=\{(\iota,\jmath)\in I_n: \iota,\jmath\not\in \{i\}\}$, $I_n^{(i,j)}=\{(\iota,\jmath)\in I_n: \iota,\jmath\not\in\{ i,j\}\}$. 
{  For a given $\theta$, define
\begin{align*}
&\hat \theta={n\choose 2}^{-1}\sum_\ijI \frac{1}{h_n}K\left(\frac{x-\Xij}{h_n}\right),\quad \hat \theta^{(i)}={n-1\choose 2}^{-1}\sum_{(k,l)\in \Ink^{(i)}}\frac{1}{h_n} K\left(\frac{x-X_{kl}}{h_n}\right), \: i=1,...,n,\\
&\hat \theta^{(i,j)}={n-2\choose 2}^{-1}\sum_{(k,l)\in \Ink^{(i,j)}}\frac{1}{h_n} K\left(\frac{x-X_{kl}}{h_n}\right),\quad  1\le i<j\le n,
\end{align*}
and subsequently $S(\theta)=\hat \theta-\theta$, $S^{(i)}(\theta)=\hat \theta^{(i)}-\theta$, and $S^{(i,j)}(\theta)=\hat \theta^{(i,j)}-\theta$.
}
Furthermore, define the pseudo true value for JEL as
\begin{align*}
V_i(\theta)=nS(\theta)-(n-1) S^{(i)}(\theta).
\end{align*}
For modified JEL, define
\begin{align}
&Q_{ij}=\frac{n-3}{n-1}\left[nS(\theta)-(n-1)\{S^{(i)}(\theta)+S^{(j)}(\theta)\}+(n-2)S^{(i,j)}(\theta)\right],\quad i<j,\nonumber\\
&\Gamma^2 = \frac{1}{n}\sumi V_i(\theta)^2,\qquad  \Gamma_m^2= \frac{1}{n}\sumi V_i(\hat \theta)^2 - \frac{1}{n}\sum_{i=1}^{n-1}\sum_{j = i+1}^n Q_{ij}^2 \label{eq:jack_var}.
\end{align}
For each value of $\theta$, define the pseudo true value for modified JEL by
\begin{align*}
 V_i^m(\theta)=V_i(\hat\theta)-\Gamma \Gamma_m^{-1}\{V_i(\hat\theta)-V_i(\theta)\}.
\end{align*}
Now, define the modified JEL function for $\theta$ by
\begin{align*}
\ell^m (\theta)&=-2\sup_{w_1,...,w_n} \sumi \log(n w_i),\\
&\text{s.t.}\quad w_i\ge 0, \quad \sumi w_i=1,\quad \sumi w_i V_i^m (\theta)=0,
\end{align*}
and $\ell(\theta)$ is defined analogously with $V_i^m$ replaced by $V_i$.
The Lagrangian dual problems are
\begin{align*}
\ell(\theta)=2\sup_\lambda \sumi \log( 1+\lambda V_i(\theta) ),\qquad \ell^m(\theta)=2\sup_\lambda \sumi \log( 1+\lambda V^m_i(\theta) ).
\end{align*}

{  We say $f$ is non-degenerate at $x$ if $\Var\left(f_{X_{12}|U_1}(x|U_1)\right)\ge\underline L>0$, which means that the vertex-specific shock $U_1$ affects the conditional density $f_{X_{12}|U_1}$. Before stating the next result, which characterizes the asymptotics of JEL and modified JEL for inference, we make the following assumptions.} 
\begin{assumption}[JEL {and modified JEL} for complete data]\label{a:JEL_complete}
	Suppose that
	\begin{enumerate}[(i)]
		\item $f(x)>0$ at the design point of interest $x$ {  which lies in $\calX\subsetneq\supp(X_{12})$}.
		\item $K$ has bounded support, $nh_n^2\to \infty$ and {  $nh_n^{5/2}\to 0$}.
	\end{enumerate}
\end{assumption}
{ 
Assumption \ref{a:JEL_complete} requires the design point $x$ to be an interior point of the support of $X_{12}$.  Assumption \ref{a:JEL_complete} (ii) imposes constraints on the convergence rate of the sequence of bandwidths. Note that the lower bound $nh_n^2\to \infty $ is a sufficient condition for Lemma 2 in the appendix, which is subsequently used for bounding the linearization errors for the JEL functions $\ell$ and $\ell^m$. This is specific to empirical likelihood-based methods and is not required by procedures such as the Wald test with FG variance estimator proposed in \cite{GrahamNiuPowell2019}. The restriction $nh_n^{5/2}\to 0$ here is used only in the degenerate case and can be relaxed to $nh_n^4\to 0$ in the non-degenerate case. 
 Note that the bandwidth conditions accommodate the MSE optimal rate of $h_n=O(n^{-2/5})$ if the researcher is only concerned with the non-degenerate case.}

\begin{theorem}[Wilks' theorem for {  JEL and} modified JEL  for dyadic KDE with complete data]\label{thm:asymptotic_dist}
	Suppose Assumptions \ref{a:rates} and \ref{a:JEL_complete} are satisfied, then
	\begin{align*}
	\ell^m(\theta)\stackrel{d}{\to} \chi_1^2.
	\end{align*}
	In addition, { 
	\begin{align*}
	\ell(\theta)\stackrel{d}{\to}\begin{cases}
\chi_1^2, \text{ if  $f$ is non-degenerate at $x$,}\\
\frac{1}{2}\chi_1^2, \text{ if  $f$ is degenerate at $x$.}
	\end{cases} 
	\end{align*}}
\end{theorem}
A proof can be found in Section D.2 of the Supplementary Appendix.

\begin{remark}[Asymptotic pivotality]\label{rem:pivotal}
Theorem \ref{thm:asymptotic_dist} shows that the modified JEL is pivotal regardless of whether $f$ is degenerate or not, while JEL is asymptotically pivotal only if $f$ is non-degenerate at $x$.  Theorem \ref{thm:asymptotic_dist} implies that one can construct an approximate $1-\alpha$ confidence interval { 
\begin{align*}
\mathcal R_\alpha=\{\theta: \ell^\bullet(\theta)\le c_\alpha \}
\end{align*} 
for $\ell^\bullet\in\{\ell,\ell^m\}$},
where $c_\alpha$ is such that $\lim_{n\to \infty}P(\theta\in \mathcal R_\alpha)=P(\chi_1^2\le c_\alpha)=1-\alpha$.
\end{remark}

{ 
\begin{remark}[Conservatism of JEL]\label{rem:JEL}
An implication of Theorem \ref{thm:asymptotic_dist} is that under degeneracy, JEL is asymptotically conservative. Thus the tests and confidence intervals based on JEL are, although not always asymptotically precise, still asymptotically valid. On the other hand, although the mJEL is asymptotically precise, simulations Section \ref{sec:simulation} shows that it is likely oversized when sample size is small. As such, JEL is still a practical option with its simpler implementation, especially if one prefers to be more conservative with the size control. 
\end{remark}
}

{ 
	\begin{remark}[Uniform inference]\label{rem:uniform_inference}
	Theorem \ref{thm:asymptotic_dist} focuses on pointwise inference. In \cite{cattaneo2022}, the authors pioneer uniform inference using strong approximation and boundary-adaptive kernels. They further develop a robust bias-correction procedure based on the bias-correction idea of \cite{calonico2018effect}. Their results are based on $t$-statistics. How to generalize a empirical likelihood type procedure to provide uniform inference for dyadic KDE remains an open question.  
	\end{remark}
}

\begin{remark}[Modified jackknife variance estimator]\label{rem:jack_var}
	A direct implication of the proof of Theorem \ref{thm:asymptotic_dist} is the consistency of $\Gamma_m^2$ in (\ref{eq:jack_var}), a bias-corrected
	jackknife variance estimator of \cite{EfronStein1981} adapted to our context.
	Based on this variance estimator, one can construct an alternative  approximate $1-\alpha$ confidence interval
 for $\theta$ as $\left[\hat{\theta}\pm  {  n^{-1/2}}z_{\alpha/2}\Gamma_m\right]$, where $z_{\alpha/2}$ is the $(1-\alpha/2)$-th quantile of the standard normal random variable. This is formalized in the following corollary.
\end{remark}
\begin{corollary}[Asymptotic normality with modified jackknife variance estimator]\label{thm:jk_var}
	Suppose Assumptions \ref{a:rates} and \ref{a:JEL_complete} are satisfied, then
\begin{align*}
\sqrt{n}\Gamma_m^{-1}(\hat\theta - \theta)\stackrel{d}{\to} N(0,1).
\end{align*}
\end{corollary}

\subsection{Inference for incomplete dyadic data with JEL and modified JEL}\label{sec:incomplete_data}
Let us now consider dyadic KDE with randomly missing incomplete data. 
Suppose the researcher observes
\begin{align*}
X_{ij}^*=Z_{ij}X_{ij},
\end{align*}
where the random variables $Z_{ij}\stackrel{d}{=}\text{Bernoulli}(p_n)$ that determine whether each edge is observed are i.i.d. and assumed to be independent from $(X_{ij})_{(i,j)\in \ijI}$ for some unknown probability of observation $p_n\in(0,1)$ { which can be but is not necessarily fixed in sample size.} Define $\hat N=\hat p_n {n\choose 2}$, $\hat  N_1=\hat p_n {n-1\choose 2}$, and $\hat N_2=\hat p_n {n-2\choose 2}$, with $\hat p_n={n\choose 2}^{-1}\sum_{\ijI} Z_{ij}$ being an estimate for $p_n$. Now let $\mathbb I_n=\{(i,j)\in \Ink:Z_{ij}=1\}$,  $\mathbb I_n^{(k)}=\{(i,j)\in \Ink:Z_{ij}=1, i,j\ne k\}$ and $\mathbb I_n^{(k,\ell)}=\{(i,j)\in \Ink:Z_{ij}=1, i,j\not\in\{ k,\ell\}\}$. {  For a fixed $\theta=f(x)$, define
the incomplete dyadic KDE estimator and its leave-out counterparts by
\begin{align*}
&\hat \theta_{\text{inc}}=\frac{1}{\hat N}\sum_{(i,j)\in \mathbb I_n} \frac{1}{h_n}K\left(\frac{x-\Xij}{h_n}\right),\quad \hat \theta_{\text{inc}}^{(i)}=\frac{1}{\hat N_1}\sum_{(k,l)\in\mathbb I_n^{(i)}}\frac{1}{h_n} K\left(\frac{x-X_{kl}}{h_n}\right),\: i=1,...,n,\\
&\hat \theta_{\text{inc}}^{(i,j)}=\frac{1}{\hat N_2}\sum_{(k,l)\in \mathbb I_n^{(i,j)}}\frac{1}{h_n} K\left(\frac{x-X_{kl}}{h_n}\right),\quad 1\le i<j\le n,
\end{align*}
and the JEL pseudo true value for incomplete dyadic data by
\begin{align*}
\hat V_i(\theta)=n\hat S(\theta)-(n-1) \hat S^{(i)}(\theta),
\end{align*}
where $\hat S(\theta)=\hat \theta_{\text{inc}}-\theta$, $\hat S^{(i)}(\theta)=\hat \theta_{\text{inc}}^{(i)}-\theta$.}
In addition, for each $i< j$, define 
\begin{align}
&\hat S^{(i,j)}(\theta)=\hat \theta_{\text{inc}}^{(i,j)}-\theta,\nonumber\\
&\hat Q_{ij}=\frac{n-3}{n-1}\left[n\hat S(\theta)-(n-1)\{\hat S^{(i)}(\theta)+\hat S^{(j)}(\theta)\}+(n-2)\hat S^{(i,j)}(\theta)\right],\nonumber\\
&\hat\Gamma^2 = \frac{1}{n}\sumi \hat V_i(\theta)^2,\qquad  \hat\Gamma^2_m= \frac{1}{n}\sumi \hat V_i(\hat \theta_{\text{inc}})^2 - \frac{1}{n}\sum_{i=1}^{n-1}\sum_{j=i+1}^n \hat Q_{ij}^2. \nonumber
\end{align}
For each $\theta$, define the modified JEL pseudo true value for incomplete dyadic data by
\begin{align*}
\hat V_i^m(\theta)=\hat V_i(\hat\theta_{\text{inc}})-\hat \Gamma \hat \Gamma_m^{-1}\{\hat V_i(\hat\theta_{\text{inc}})- \hat V_i(\theta)\}.
\end{align*}
Now the modified JEL function for incomplete dyadic data for $\theta$ can be defined by
\begin{align*}
\hat\ell^m(\theta)&=-2\sup_{w_1,...,w_n} \sumi \log(n w_i),\\
&\text{s.t.}\quad w_i\ge 0, \quad \sumi w_i=1,\quad \sumi w_i \hat V_i^m (\theta)=0.
\end{align*}
and $\hat\ell(\theta)$ is defined analogously with $\hat V_i^m$ replaced by $\hat V_i$.
The Lagrangian dual problems are now
\begin{align*}
 \hat\ell(\theta)=2\sup_\lambda \sumi \log( 1+\lambda \hat V_i(\theta) ),\qquad 
 \hat\ell^m(\theta)=2\sup_\lambda \sumi \log( 1+\lambda \hat V^m_i(\theta) ).
\end{align*}

\begin{assumption}[JEL and modified JEL for incomplete data]\label{a:JEL_incomplete}
	Suppose $Z_{ij}\stackrel{d}{=}\text{Bernoulli}(p_n)$ are i.i.d. independent from $(X_{ij})_{(i,j)\in \ijI}$ for some unknown sequence $p_n\in(0,1)$ with $nh_n p_n  \to \infty$.
\end{assumption}
Assumption \ref{a:JEL_incomplete} imposes the random missing structure on the data and establishes a lower bound on the rate of $p_n$, the probability of observing an edge, to ensure the presence of enough data to establish asymptotic theory. Note that the this assumption accommodates sequences of $p_n$ that converge to zero slowly enough. { The condition $nh_n p_n\to \infty$ ensures that we have asymptotically increasing effective sample size. }

The following result is an incomplete data counterpart of Theorem \ref{thm:asymptotic_dist}.
\begin{theorem}[Wilks' theorem for { JEL and} modified JEL for dyadic KDE with incomplete data]\label{thm:asymptotic_dist_incomplete}
	Suppose Assumptions \ref{a:rates}, \ref{a:JEL_complete} and \ref{a:JEL_incomplete} are satisfied, then
\begin{align*}
\hat\ell^m(\theta)\stackrel{d}{\to} \chi_1^2.
\end{align*}
In addition, if  $f$ is non-degenerate at $x$ or $p_n\to 0$, then
\begin{align*}
\hat\ell(\theta)\stackrel{d}{\to}\chi_1^2.
\end{align*}
\end{theorem}
A proof can be found in Section D.3 of the Supplementary Appendix.

{ 
Theorem \ref{thm:asymptotic_dist_incomplete} states the asymptotic distributions of the proposed JEL and modified JEL statistics for incomplete data. As in Theorem \ref{thm:asymptotic_dist}, the modified JEL is asymptotically pivotal regardless of the asymptotic regime. However, with incomplete data, the asymptotic pivotality of the proposed JEL statistic happens not only under nondegeneracy but also when the proportion of observed data goes to zero slowly asymptotically. {  In such a scenario, the term that consists of the randomness induced by the missing process (which is conditionally independent) dominates the original leading asymptotic term, which is potentially not pivotal. }  Note that the construction of confidence intervals in Remark \ref{rem:pivotal} can be adapted with { $\ell^\bullet\in\{\hat \ell^m,\hat \ell\}$}. }

\section{Simulation studies}\label{sec:simulation}
We consider four sets of data generating processes (DGP) in our simulations. The first two follow
\begin{align*}
X_{ij}=\beta U_i U_j+U_{\{i,j\}}
\end{align*}
with $U_{\{i,j\}}\stackrel{i.i.d.}{\sim}N(0,1)$ and $U_i=-1$ with probability $1/3$ and equals $1$ otherwise. We set $\beta\in\{0,1\}$, where $\beta=1$ is the sufficiently non-degenerate DGP considered in \cite{GrahamNiuPowell2019}. This has the density:
\begin{align*}
f(x)=\frac{5}{9}\phi(x-1)+\frac{4}{9}\phi(x+1),
\end{align*}
where $\phi$ is the density of a standard normal random variable. Meanwhile, $\beta=0$ corresponds to the degenerate case where the true DGP is i.i.d. standard normal with density $f=\phi$. {  The last two sets of DGPs are the incomplete data counterparts of the first two, with the same DGPs except that not all of the edges are observed following the set-up in Section \ref{sec:incomplete_data}. For these cases, we set $p_n=0.5$.}
We utilize the rule of thumb bandwidths $h_n^{S}$ and $h_n^{S,\text{inc}}$ proposed in (A.1) and (A.2) in Section A of the Supplementary Appendix, respectively. For complete dyadic data, we consider five alternative inference methods that are theoretically robust to dyadic clustering: (i) Wald statistic with FG variance estimator for dyadic KDE proposed in \citet[Equation (22)]{GrahamNiuPowell2019} (FG), (ii) Wald statistic with the leading term of FG (lFG), (iii) jackknife empirical likelihood (JEL), (iv) Wald statistic with modified jackknife variance estimator (mJK) from Remark \ref{rem:jack_var}, and (v) modified jackknife empirical likelihood (mJEL). Note that (iii)-(v) are studied in this paper. {  For incomplete dyadic data, we consider the incomplete counterparts of (iii) and (v) from Section \ref{sec:incomplete_data}.}
Throughout the simulation studies, we set the design point $x=1.675$, consistent with the simulation study in \cite{GrahamNiuPowell2019}. Each simulation is iterated $5,000$ times.\\
\par

Tables \ref{table:dyadic_cov_prob}, \ref{table:iid_cov_prob}, \ref{table:incomplete_dyadic}, and \ref{table:incomplete_iid} show the simulation results under these different settings.
In general, the simulation supports our theoretical findings. {  For complete data, under both the dyadic and i.i.d. DGPs, both mJEL and mJK-based confidence intervals enjoy close to nominal coverage rates even with moderate sample sizes, while the FG estimator is oversized with small sample sizes and converges to the nominal size when $n$ is large. The mJK-based intervals are slightly more oversized than mJEL-based intervals for small sample sizes but they perform similarly when $n$ grows larger. The original JEL and lFG estimators are conservative but asymptotically valid under the dyadic DGP. Consistent with Remark 5, the JEL estimator is always severely conservative under the i.i.d. DGP, as is the lFG estimator. The behaviors of JEL and mJEL-based intervals under incomplete data are similar to those we observed with complete data; JEL is conservative but asymptotically converging to the correct size in the non-degenerate case while mJEL remains quite precise under both degenerate and non-degenerate DGPs.}


\section{Empirical application: Airport congestion}\label{sec:application}
Flight delays caused by airport congestion are a familiar experience to the flying public. Aside from the frustration experienced by affected passengers, these delays are an important cause of fuel wastage and result in the release of criteria pollutants such as nitrogen oxides, carbon monoxide and sulfur oxides. As discussed in \cite{SchlenkerWalker2016}, these releases harm the health of people living near airports. To mitigate these harms, it is important to understand the distribution of congestion across routes, or flights between pairs of airports.\\
\par
In this empirical application, we apply the dyadic KDE and JEL and modified JEL procedure to understand the distribution of airport congestion, as measured by taxi time, across pairs of airports in the United States. The taxi time of a flight is defined as the sum of taxi-out time (the time between the plane leaving its parking position and taking off) and taxi-in time (the time between the plane landing and parking). In this context, the edges $X_{ij}$ and $X_{jk}$ represent flights between airport $i$ and airport $j$ and flights between airport $j$ and airport $k$ respectively; the fact that flights share a common airport $j$ induces dependence between $X_{ij}$ and $X_{jk}$. Therefore, this is a natural setting to apply dyadic KDE and modified JEL.\\
\par
Data on flight taxi times are obtained from the US DOT Bureau of Transportation Statistics (BTS) \textit{Reporting Carrier On-Time Performance} dataset, which compiles data on US domestic non-stop flights from all carriers which earn at least 0.5\% of scheduled domestic passenger revenues. From the flight-level data for June 2020, we collapse taxi time into various summary statistics (mean, 95th percentile, maximum) for each origin-destination airport pair, with the summary statistics being taken over flights in both directions so that the resulting network is undirected. Noting that a missing edge means that there were no non-stop flights between the vertices the edge connects, we apply the KDE, JEL, and modified JEL to the resulting network. For simplicity, the sample is restricted to the 100 largest airports by number of departing flights.\\
\par
The histograms, dyadic KDE estimates, pointwise JEL and pointwise modified JEL 95\% confidence intervals for the mean, 95th percentile and maximum taxi times between airport pairs in June 2020 are shown in Figure 1A-1C. The intervals obtained by numerically inverting the test statistic for each design point (each whole minute in the range of each statistic). {  Since both the JEL and modified JEL confidence intervals in Figure 1 are pointwise, a comparison between JEL and mJEL can only be made for each of the finitely many points at which the confidence intervals were calculated. Note that the design points are the same for the JEL and mJEL graphs, making a point-by-point comparison possible.}\\
\par
For applied researchers, the results demonstrate how studying only the mean can be incomplete; unlike the mean, the 95th percentile and (particularly) maximum travel times exhibit positive skewness. Routes with taxi times lying in the positively skewed region cause the most problems, but this distinction is lost when only studying means. For all these statistics, the modified JEL procedure which we propose provides relatively precise estimates even under small sample sizes (100 vertices) and a large proportion of missing edges (72\% missing).

\section{Conclusion}
This paper studies the asymptotic properties of dyadic kernel density estimation. We establish uniform convergence rates under general conditions. For inference, we propose a modified jackknife empirical likelihood procedure which is valid regardless of degeneracy of the underlying DGP. We further extend the results to cover incomplete or missing at random dyadic data. Simulation studies show robust finite sample performance of the modified JEL inference for dyadic KDE in different settings. We illustrate the method via an application to airport congestion.\\
\par
Despite our focus on density estimation, the inference approach we take can be extended to cover the local regression as in \cite{graham2021minimax}, the local polynomial density estimation of \cite{CattaneoMaJansson2020JASA} or local polynomial distributional regression of \cite{CattaneoMaJansson2020wp} under dyadic data. Another potential direction is to study incomplete data under unconfoundedness or other non-random missing models. 
{  Finally, in light of the recent paper of \cite{cattaneo2022} which pioneers uniform inference for dyadic KDE using Wald statistics, it would be interesting to consider ways to adapt modified JEL for uniform inference as well.  }
Investigating the modified JEL under these assumptions provides interesting directions for future research.

\clearpage

\appendix

\allowdisplaybreaks

\section{Practical guideline for bandwidth}\label{sec:practical_guideline}

Here we provide some rules of thumb for bandwidth choices. As mentioned in Section \ref{sec:theory}, the precise bandwidth is not as important in the non-degenerate case since it does not enter the first-order asymptotics. 
As shown in \cite{GrahamNiuPowell2019}, the MSE-optimal bandwidth has the expression
\begin{align*}
h_n^{MSE}=\left\{\frac{1}{4} \frac{f(w)c_1 }{\left(2^{-1}f''(w) c_2\right)^2}\frac{2}{n(n-1)}\right\}^{1/5},
\end{align*}
where $c_1=\int K^2(u)du$ and $c_2=\int u^2 K(u)du$. \cite{GrahamNiuPowell2019} pointed out that, under non-degeneracy, KDE is robust to different rates of bandwidths; the unknowns $f(w)$ and $f''(w)$ do not affect its asymptotic behavior. {  On the other hand, under degeneracy, for the squared bias to be negligible to variance, the theory requires both $nh_n^2\to\infty $ and $nh_n^{5/2}\to 0$. Hence we set $h_n\sim n^{-4/9}$ }and propose {  the following constant, which is an adaptation of the rule of thumb by \citet[pp. 48, Equation (3.31)]{Silverman1986}}
\begin{align}
h_n^{S}=0.9\cdot \min\left\{\hat\sigma,\text{IQR}/1.34\right\}\cdot \left\{\frac{2}{n(n-1)}\right\}^{2/9},\label{eq:ROT_bandwidth}
\end{align}
where $\hat\sigma$ and $\text{IQR}$ are the standard error and the interquartile range of $(X_{ij})_{(i,j)\in \Ink}$, i.e. the difference of the $75$-th quantile and the $25$-th quantile, respectively.
For incomplete data, to account for the effective sample size, we propose
\begin{align}
h_n^{S,\text{inc}}=0.9\cdot \min\left\{\hat\sigma_{\text{inc}},\text{IQR}_{\text{inc}}/1.34\right\}\cdot \left\{\frac{2}{\hat p \cdot n(n-1)}\right\}^{2/9},\label{eq:ROT_bandwidth_incomplete}
\end{align}
where $\hat\sigma_{\text{inc}}$ and $\text{IQR}_{\text{inc}}$ are the standard error and the interquartile range of $(X_{ij})_{(i,j)\in \mathbb I_n}$, respectively. As discussed in Section \ref{sec:simulation}, these rule of thumb bandwidth choices work well in our simulation studies.
Formal analysis of the optimality of coverage probabilities of different bandwidth choices, such as those studied in \cite{CalonicoCattaneoFarrell2018}, under dyadic asymptotics would be an interesting future venue of research. 
\section{Some empirical process definitions and results}\label{sec:empirical process}
We shall first list some useful definitions from the empirical process theory literature. For more detail, see \cite{vdVW1996} or \cite{GineNickl2016}.
Let $S$ be a set and $\calC$ be a nonempty class of subsets of $S$. Pick any finite set $\{x_1,...,x_n\}$ of size $n$. We say that $\calC$ picks out a subset $A\subset \{x_1,...,x_n\}$ if there exists $C\in \calC$ such that $A=\{x_1,...,x_n\}\cap C$.
Let $\Delta^\calC(x_1,...,x_n)$ be the number of subsets of $\{x_1,...,x_n\}$ picks out by $\calC$.
We say the class $\calC$ shatters $\{x_1,...,x_n\}$ if $\calC$ picks out all of its $2^n$ subsets.  
The VC index $V(\calC)$ is defined by the smallest $n$ for which no set of size $n$ is shattered by $\calC$, i.e., with $m^\calC(n):=\max_{x_1,...,x_n}\Delta^\calC(x_1,...,x_n)$,
\begin{align*}
V(\calC)=\begin{cases}
\inf\{n:m^\calC(n)<2^n\}, \text{ if such set is non-empty,}\\
+\infty, \text{ otherwise.}
\end{cases}
\end{align*}
The class $\calC$ is called a Vapnik–Chervonenkis class, or VC class for short, if $V(\calC)<\infty$. The subgraph of a real function $f$ on $S$ is defined as 
$
\sg(f)=\{(x,t):t<f(x)\}
$.
A function class $\calF$ on $S$ is a VC subgraph class if the collection of subgraphs of $f\in\calF$, $\sg(\calF)=\{\sg(f):f\in\calF\}$, is a VC class of sets in $S\times \Real$. We can define the VC index of $\calF$, $V(\calF)$, as the VC index of $\sg(\calF)$. 
Let $(T,d)$ be a pseudometric space (i.e. $d(x,y)=0$ does not imply $x=y$). For $\varepsilon>0$, an $\varepsilon$-net of $T$ is a subset $T_\varepsilon$ of $T$ such that for every $t\in T$ there exists $t_\varepsilon\in T_\varepsilon$ with $d(t,t_\varepsilon)\le \varepsilon$. The $\varepsilon$-covering number $N(T,d,\varepsilon)$ of $T$ is defined by
\begin{align*}
N(T,d,\varepsilon)=\inf \{\text{Card}(T_\varepsilon):\, T_\varepsilon \text{ is an $\varepsilon$-net of $T$}\}.
\end{align*}
For a probability measure $Q$ on a measurable space $(S,\calS)$,
for any $r\ge 1$, denote $\|f\|_{Q,r}=\left(\int |f|^r dQ\right)^{1/r}$ for $f\in L^r(S)$. A function $G:S\to\R$ is an envelope of a class of functions $\calG\ni g$, $g:S\to\R$, if $\sup_{g\in\calG}|g(s)|_\infty\le G(s)$ for all $s\in S$.
Suppose $\calF$ is a VC subgraph class with envelope $F$, then Theorem 2.6.7 in \cite{vdVW1996} suggests that for any $r\in [1,\infty)$, we have
\begin{align*}
\sup_{Q} N(\calF,\|\cdot\|_{Q,r},\varepsilon\|F\|_{Q,r}) \le KV(\calF)(16e)^{V(\calF)} \left(\frac{2}{\varepsilon}\right)^{r (V(\calF) -1)}
\end{align*}
for all $0<\varepsilon\le 1$, where $K$ is universal and the supremum is taken over all finite discrete probability measures.
\par
A closely related concept is the VC type. We say a set of functions $\calF$ with envelope $F$ is a VC type class with characteristics $(A,v)$ if for some positive constants $A,v$,
\begin{align}
\sup_{Q}N(\calF,\|\cdot\|_{Q,2},\varepsilon\|F\|_{Q,2})\le \left(\frac{A}{\varepsilon}\right)^v,\label{eq:VC-type}
\end{align}
where the supremum is taken over all finite discrete measures on $(S,\calS)$.\\
\par
The following restates Lemma A.2 in \cite{GhosalSenvdV2000}, a useful result for considering uniform covering numbers for conditional expectations of classes of functions.
\begin{lemma}\label{lem:conditional_entropy}
	Let $\calH$ be a class of functions $h:\calX\times \calY\to \R$ with envelope $H$ and $R$ a fixed probability measure on $\calY$. For a given $h\in \calH$, let $\overline h:\calX\to \R$ be $\overline h=\int h(x,y)dR(y)$. Set $\overline \calH=\{\overline h:h\in\calH\}$. Note that $\overline H$ is an envelope of $\overline \calH$. 
	Then, for any $r,s\ge 1$, $\varepsilon\in(0,1]$,
	\begin{align*}
	\sup_{Q}N(\overline\calH,\|\cdot\|_{Q,r},2\varepsilon\|\overline H\|_{Q,r})\le 	\sup_{Q}N(\calH,\|\cdot\|_{Q\times R,s},\varepsilon^r\|H\|_{Q\times R,s})	.
	\end{align*}

\end{lemma}

\section{Auxiliary lemmas}

{ 
Throughout the rest of this Supplementary Appendix, denote
	\begin{align*}
	\Omega_1(x)=\Var\left(f_{X_{12}|U_1}(x|U_1)\right),\qquad\Omega_2(x)=f(x)\int K^2(u) du.
	\end{align*}
	Also, for notational convenience, define $K_{ij,n}=K_{ji,n}$ for $1\le j<i \le n$.
\begin{lemma}\label{lem:max_V}
	Under Assumptions \ref{a:rates} and \ref{a:JEL_complete}, it holds that
	$\max_{1\le i \le n}|V_i(\theta)|=O_P(h_n^{-1})$ under non-degeneracy and $\max_{1\le i \le n}|V_i(\theta)|=O_P(h_n^{-1/2})$ under degeneracy. 
	\begin{proof}
		Notice that a direct calculation yields that for each $i=1,...,n$,
		\begin{align*}
		S(\theta)= \frac{2}{n(n-1)}\sum_{\ijI}K_{ij,n}-\theta=\frac{2}{n(n-1)}\sum_{j\ne i}^n(K_{ij,n}-\theta) + \frac{n-2}{n}S^{(i)}(\theta).
		\end{align*}
		Hence under non-degeneracy at $x$, one has
		\begin{align*}
		|V_i(\theta)|=&\left|\frac{2}{(n-1)}\sum_{j\ne i}^n(K_{ij,n}-\theta) -S^{(i)}(\theta)\right|
		\lesssim \max_{1\le i< j\le n}|K_{ij,n}|.
		\end{align*}
		To bound the right hand side with high probability, notice that as $nh_n^2\to \infty$, we have
		\begin{align*}
		&\E\left[\max_{1\le i< j\le n}|K_{ij,n}|\right]=\frac{1}{h_n}\E\left[\max_{1\le i< j\le n}K\left(\frac{x-X_{ij}}{h_n}\right)\right]= O(h_n^{-1}).
		\end{align*}

	Now suppose the DGP is degenerate at $x$, then by Jensen's inequality
\begin{align*}
E\left[\max_{1\le i \le n}|V_i(\theta)|\right]\le& \left(E\left[\max_{1\le i \le n}V_i(\theta)^2\right]\right)^{1/2}
\le \left(E\left[\sum_{i=1}^n V_i(\theta)^2\right]\right)^{1/2}\\
\le& n^{1/2}  \left(E\left[\frac{1}{n}\sum_{i=1}^nV_i(\theta)^2\right]\right)^{1/2}
= n^{1/2}  \left(\frac{4\Omega_2(x)(1+o(1))}{nh_n}\right)^{1/2}=O(h_n^{-1/2}),
\end{align*}
where the second to the last equality follows from Lemma \ref{lem:V2}.
	\end{proof}
\end{lemma}
}

\begin{lemma}\label{lem:third_moment}
	Under Assumptions \ref{a:rates} and \ref{a:JEL_complete}, we have $n^{-1}\sumi |V_i(\theta)|^3=o_P(n^{1/2})$ under non-degeneracy {  and $n^{-1}\sumi |V_i(\theta)|^3=O_P\left((nh_n)^{-3/2}\right)$ under degeneracy.}
	\begin{proof}
		Under non-degeneracy, it follows from Lemma \ref{lem:max_V}, Lemma \ref{lem:V2}, and Assumption \ref{a:JEL_complete} (ii) that
		\begin{align*}
		\frac{1}{n}\sumi |V_i(\theta)|^3\le \max_{1\le i \le n}|V_i(\theta)|\cdot\frac{1}{n}\sumi V_i(\theta)^2=O_P(h_n^{-1})\cdot O_P(1)=o_P(n^{1/2}).
		\end{align*}
		{ 
		Under degeneracy, following the decomposition  in the proof of Lemma \ref{lem:max_V}, we have
		\begin{align*}
		|V_i(\theta)|\lesssim \left|\frac{2}{(n-1)}\sum_{j\ne i}^n(K_{ij,n}-\theta)\right| +\left|S^{(i)}(\theta)\right|\lesssim O_P\left(\frac{1}{\sqrt{nh_n}}\right)+ O_P\left(\frac{1}{\sqrt{n^2h_n}}\right),
		\end{align*}
		where the last inequality follows from the $\sqrt{nh_n}$-asymptotic normality of kernel estimation for independent data, and the $\sqrt{n^2h_n}$-asymptotic normality of $S^{(i)}(\theta)=\hat \theta^{(i)}-\theta$ under degeneracy from Section 4 in \cite{GrahamNiuPowell2019}.
		Thus we have
	 $n^{-1}\sumi |V_i(\theta)|^3=O_P\left((nh_n)^{-3/2}\right)$.
		}
	\end{proof}
\end{lemma}

\begin{lemma}\label{lem:lambda_conv}
{ 	Suppose the assumptions of Theorem \ref{thm:asymptotic_dist} are satisfied, then
	$\hat \lambda=\Op(n^{-1/2})$ under non-degeneracy and $\hat \lambda=\Op(h_n^{1/2})$ under degeneracy.}
\end{lemma}
\begin{proof}[Proof of Lemma \ref{lem:lambda_conv}]
	The proof follows from the same strategy as in the proof of \citet[Equation (2.14)]{Owen1990}. The first-order condition yields
	\begin{align*}
	0=\frac{1}{n}\sumi \frac{V_i(\theta)}{1+\hat \lambda V_i(\theta)}.
	\end{align*}
	Denote $\hat\lambda=\rho\beta$, where $\rho=|\hat \lambda| \ge 0$ and $|\beta|=1$. Then
	\begin{align*}
	0=\left|\frac{1}{n}\sumi \frac{V_i(\theta)}{1+\rho\beta V_i(\theta)}\right|
	=
	\frac{1}{n}\left|\sumi V_i(\theta)-\sumi \frac{\rho\beta V_i(\theta)^2}{1+\rho\beta V_i(\theta)}\right|
	\ge \rho\frac{ \frac{1}{n}\sumi V_i(\theta)^2}{1+\rho \max_{1\le i \le n}|V_i(\theta)|}-\left|S(\theta)\right|.
	\end{align*}
	
We now consider two asymptotic regimes separately.
First, under non-degeneracy,
	note that $\left|S(\theta)\right|=O_P(n^{-1/2})$ following the uniform convergence rate from Theorem \ref{thm:unif_rate}. Further, $n^{-1}\sumi V_i(\theta)^2$ is bounded away from zero with probability $1-o(1)$  following Lemma \ref{lem:V2}. In addition, we have $\max_{1\le i \le n}|V_i(\theta)|=o_P(n^{1/2})$ following Lemma \ref{lem:max_V}. Therefore, $\rho=|\hat \lambda|=\Op(n^{-1/2})$.
	
	{  On the other hand, under degeneracy,
$\left|S(\theta)\right|=O_P((n^2h_n)^{-1/2})$ following the variance calculations in Section 3 in \cite{GrahamNiuPowell2019}. Also, $(nh_n)\cdot n^{-1}\sumi V_i(\theta)^2$ is bounded away from zero with probability $1-o(1)$ following Lemma \ref{lem:V2}. In addition, we have $\max_{1\le i \le n}|V_i(\theta)|=O_P(h_n^{-1/2})$ following Lemma \ref{lem:max_V}. Then
\begin{align*}
\frac{ \rho}{1+ O_P(\rho h_n^{-1/2})}= O_P\left(\left|S(\theta)\right|\cdot (nh_n)\right)= O_P\left(h_n^{1/2}\right).
\end{align*}
Thus we have $\rho=|\hat \lambda|=\Op(h_n^{1/2})$.
	 }
\end{proof}	


\begin{lemma}\label{lem:V2}
Under Assumptions \ref{a:rates} and \ref{a:JEL_complete}, denote $\sigma^2=4\Omega_1(x)+{ 2}\Omega_2(x)/(nh_n)$, we have,
under non-degeneracy, 
	\begin{align*}
	\frac{1}{n\sigma^2}\sumi V_i(\theta)^2= \frac{4\Omega_1(x)}{\sigma^2}+o_P(1),
	\end{align*}
	and under degeneracy,
	\begin{align*}
	\frac{1}{n\sigma^2}\sumi V_i(\theta)^2=\frac{{ 4} \Omega_2(x)}{\sigma^2 nh_n} +o_P(1).
	\end{align*}
\begin{proof}
First, note that the following algebraic facts hold:  $V_i=S-(n-1)(S^{(i)}-S)$ and
	\begin{align}
	\frac{1}{n}\sumi S^{(i)}(\theta)=\frac{2}{n(n-1)(n-2)}\sumi \left(\sum_{\iota=1}^{n-1}\sum_{j>\iota}K_{\iota j,n}-\sum_{j> i}K_{ij,n}-\sum_{\iota< i}K_{\iota i,n}\right)-\theta=S(\theta). \label{eq:leave-one-out_trick}
	\end{align}
 Following the decomposition of Lemma 3 in \cite{MatsushitaOtsu2020},
	\begin{align*}
	\frac{1}{n\sigma^2}\sumi V_i(\theta)^2=&\frac{S(\theta)^2}{\sigma^2} - 2(n-1) S(\theta)\frac{1}{n\sigma^2} \sumi \{S^{(i)}(\theta)-S(\theta)\} +(n-1)^2\frac{1}{n\sigma^2}\sumi\{S^{(i)}(\theta)-S(\theta)\}^2\\
	=&T_1-2T_2+T_3.
	\end{align*}
	Observe that $T_1=o_P(1)$ following the uniform convergence rate from Theorem \ref{thm:unif_rate} and $T_2=0$ using Equation (\ref{eq:leave-one-out_trick}). Now, following the same decomposition of (\ref{eq:Hoeffding_decomposition}) in the proof of Theorem \ref{thm:unif_rate},
	we have
	\begin{align}
	S^{(i)}=\theta^{(i)}-\theta=\frac{1}{n-1}\sum_{k\ne i}L_k+\frac{2}{(n-1)(n-2)}\sum_{k\ne i }\sum_{l>k, l\ne i}(W_{kl}+R_{kl}) +\text{Bias}(\hat \theta^{(i)}).\label{eq:Hoeffding_decomposition_leave-one-out}
	\end{align}
	In addition, notice that $ \sum_{k\ne i }\sum_{l>k, l\ne i}R_{kl}=\sum_{k=1 }^{n-1}\sum_{l= k+1}^nR_{kl}-\sum_{l> i}R_{il}-\sum_{k< i}R_{ki}.$
	Hence, using Equation (\ref{eq:leave-one-out_trick}) and the variance estimator expressed as a $U$-statistic, we have
	\begin{align*}
	T_3=&(n-1)^2\frac{1}{n\sigma^2}\sumi\{S^{(i)}(\theta)-S(\theta)\}^2\\
	=&(n-1)^2\frac{1}{n\sigma^2}\sumi\left\{S^{(i)}(\theta)-\frac{1}{n}\sumi S^{(i)}(\theta)\right\}^2\\
	=&(n-1)^2\frac{1}{n^2\sigma^2}\sum_{i=1}^{n-1}\sum_{j=i+1}^n\left\{S^{(i)}(\theta)- S^{(j)}(\theta)\right\}^2\\
	=&(n-1)^2\frac{1}{n^2\sigma^2}\sum_{i=1}^{n-1}\sum_{j=i+1}^n\left\{\frac{1}{n-1}(L_j-L_i )+\frac{2}{(n-1)(n-2)}\sum_{l\ne i, j}\{R_{jl} -R_{il}\} \right\}^2+o_P(1)\\
	=&\frac{1}{\sigma^2}\left\{\Var(L_1)+\frac{ { 4} }{n}\Var(R_{12})\right\}+o_P(1)=\frac{1}{\sigma^2}\left\{4\Omega_1(x)+\frac{ { 4} }{n h_n }\Omega_2(x)\right\}+o_P(1),
	\end{align*}
	where the third equality follows from \citet[pp. 589]{EfronStein1981}, the fourth follows from  the decomposition of Equation (\ref{eq:Hoeffding_decomposition_leave-one-out}), the fact that $\text{Bias}(\hat \theta^{(i)})=\text{Bias}(\hat \theta^{(j)})$ (as they utilize the same bandwidth) and the uniform rate of Theorem \ref{thm:unif_rate}, and the second to the last by the calculation
	\begin{align*}
&\frac{(n-1)^2}{n^2\sigma^2}\sum_{i=1}^{n-1} \sum_{j=i+1}^n \left(\frac{2}{(n-1)(n-2)}\sum_{l\ne i, j}\{R_{jl} -R_{il}\} \right)^2\\
=&
\frac{(n-1)}{2n(n-2)\sigma^2}\E\left[ { 4}(R_{23} -R_{13})^2\right]+o_P(1)
=
\frac{{ 4}}{n\sigma^2}\E\left[ \frac{1}{2}(R_{23} -R_{13})^2\right]+o_P(1)=\frac{ { 4} }{n\sigma^2}\Var(R_{12})+o_P(1),
	\end{align*}
	 following the WLLN for triangular arrays, the conditional independence (on $(U_i)_{1\le i \le n}$) of $(R_{ij})_{(i,j)\in\Ink}$,  and the law of total variance.
	Hence, under non-degeneracy, we have
	$
	n^{-1}\sigma^{-2}\sumi V_i(\theta)^2= 4\Omega_1(x)/\sigma^2+o_P(1),
	$
	and under degeneracy,
	$
		n^{-1}\sigma^{-2}\sumi V_i(\theta)^2={ 4} \Omega_2(x)/(\sigma^2nh_n)+o_P(1).
	$
\end{proof}
\end{lemma}
\section{Proofs of the main results}
Throughout this appendix, for notational convenience, for any $i>j$, let us define $X_{ij}:=X_{ji}$.
\subsection{Proof of Theorem \ref{thm:unif_rate}}
\begin{proof}\label{proof:unif_rate}
	Note that $\hat f_n(x)- f(x)=(\E[\hat f_n(x)]+f(x))+(\hat f_n(x)-\E[\hat f_n(x)])$.
	The bound for the non-stochastic component follows from the identical distributions of $X_{ij}$ and a direct calculation
	\begin{align*}
	\E[\hat f_n(x)]=\E\left[\frac{1}{h_n}K\left(\frac{X_{12}-x}{h_n}\right)\right]=\int K(u)duf(x-uh_n)du=f(x)+h_n^2B(x) + o(h_n^2).
	\end{align*}
	where
	$
	B(x):=2^{-1}f''(x) \int u^2 K(u)du
	$.
	Thus $\text{Bias}(\hat f_n(x)):=\E[\hat f_n(x)-f(x)]=h_n^2B(x)+o(h_n)=O(h_n^2)$.	
	
	We shall now consider the stochastic component.	
	For each $n\in\N$, write $\calG_n=\{x'\mapsto K((x-x')/h_n): x\in \R\}$. Then, following the bias calculation above, one has $\sup_{g\in\calG_n}\E[g^2(X_{12})]\le \|f\|_\calX h_n+O(h_n^3)=O( h_n)$. and for each $x\in\R$, there exists a $g\in \calG_n$ such that 
	\begin{align*}
	\hat f_{n}(x)={n\choose 2}^{-1}\sum_{\ijI}\frac{1}{h_n}\cdot g(X_{ij}).
	\end{align*}
	Denote $g_1(\cdot)=\E[g(X_{12})|U_1=\cdot]$.
	The decomposition in \cite{GrahamNiuPowell2019} (Equations (9)-(11)) yields
	\begin{align}
	&{n\choose 2}^{-1}\sum_{\ijI} \frac{1}{h_n}\cdot (g(X_{ij})-\E[g(X_{12})])\nonumber\\
	=&\frac{2}{n}\sumi \frac{1}{h_n}(g_1(U_i) -\E[g(X_{12})])\nonumber\\
	&+\frac{2}{n(n-1)}\sum_{\ijI}\frac{1}{h_n}\left(\E[g(X_{ij})|U_i,U_j] -g_1(U_i) -g_1(U_j)+\E[g(X_{12})]\right)\nonumber\\
	&+\frac{2}{n(n-1)}\sum_{\ijI}\frac{1}{h_n}\left(g(X_{ij})-\E[g(X_{ij})|U_i,U_j]\right)\nonumber\\
	=&\frac{1}{n}\sumi L_i(g) +\frac{2}{n(n-1)}\sum_\ijI W_{ij}(g)+\frac{2}{n(n-1)}\sum_\ijI  R_{ij}(g).\label{eq:Hoeffding_decomposition} 
	\end{align}
	Note that each of the three terms on the right hand side is centered. It suffices to bound each terms on the right hand side of (\ref{eq:Hoeffding_decomposition}).\\
	\par
	We shall first obtain some preliminary results for the uniform covering numbers of the classes of functions in these terms.
	Observe that $L_i(g)=2\{\int K(z) f_{Y_{12}|U_1}(x+zh_n| U_i) dz-E[\int K(z) f_{Y_{12}|U_1}(x+zh_n| U_i) dz]\}$, where $g\in \calG_n$ corresponds to the design point $x$. Now, define the classes of functions with support $(z,u)\in\supp(K)\times [0,1]$ by
	\begin{align*}
	\calG_{2n}=&\left\{(z,u)\mapsto x+zh_n : x\in \calX\right\},\quad
	\calG_{3}=\left\{(z,u)\mapsto u \right\},\\
	\calG_{4n}=&\left\{(z,u)\mapsto f_{Y_{12}|U_1}(x+zh_n\mid u)  : x\in \calX\right\},\\
	\overline \calG_{4n}=&\left\{u\mapsto \int K(z) f_{Y_{12}|U_1}(x+zh_n| u) dz : x\in \calX\right\}.
	\end{align*} 
	Note that $\calG_{2n}$ is a subset of a two dimensional vector space with an envelope $M:=\text{radius}(\supp(K) )+\text{radius}(\calX )$. Thus by Lemma 23 in \cite{Kato2017lecture}, it is a VC-subgraph with VC-index $\le 5$ and hence there exist universal constants $A\ge 1$ and $v\ge 1$ that for any $\epsilon\in (0,1]$,
	\begin{align*}
	\sup_{Q}N(\calG_{2n} ,\|\cdot\|_{Q,2},\epsilon M)\le (A/\epsilon)^{5v}.
	\end{align*}
	Now, $\calG_{3}$ is a class consists of one single function with envelope $1$ and thus for any $\epsilon\in (0,1]$, one has $
	\sup_{Q}N(\calG_{3} ,\|\cdot\|_{Q,2},\epsilon )= 1$. Now, note that for any $(g_1,g_2), (g_1',g_2')\in \calG_{2n}\times \calG_{3}$ and $\phi(g_1,g_2)(z,u):=f_{Y_{12}|U_1}(g_1(z,u)\mid g_2(z,u))$, we have by Assumption \ref{a:rates} (ii) that
	\begin{align*}
	|\phi(g_1,g_2)(z,u)-\phi(g_1',g_2')(z,u)|^2\le \|f_{Y_{12}|U_1}'\|_\infty^2 |g_1(z,u)-g_1'(z,u)|^2=O(|g_1(z,u)-g_1'(z,u)|^2).
	\end{align*}
	Thus by Proposition 5 in \cite{Kato2017lecture}, it holds that
	\begin{align*}
		\sup_{Q}N(\calG_{4n} ,\|\cdot\|_{Q,2},\epsilon M)=	\sup_{Q}N(\phi(\calG_{2n},\calG_3) ,\|\cdot\|_{Q,2},\epsilon M)\lesssim (A/\epsilon)^{5v}.
	\end{align*}
	An application of Lemma \ref{lem:conditional_entropy} with $R= K$, $r=s=2$ then yields that for any $\epsilon\in (0,1]$, 
\begin{align*}
\sup_{Q}N(\overline \calG_{4n} ,\|\cdot\|_{Q,2},\epsilon M)\le \sup_{Q}N(\calG_{4n} ,\|\cdot\|_{Q,2},\epsilon^2 M/4) \lesssim (A/\epsilon)^{5v}=O(1).
\end{align*}
	Now, we come back to the first term of (\ref{eq:Hoeffding_decomposition}). First, note that this term is zero in the degenerate case. We then consider the non-degenerate case. Note that the summands are i.i.d. Thus using the entropy bounds from above and an application of Theorem 5.2 in \cite{CCK2014AoS} yield that
\begin{align*}
\E\left[\sup_{g\in\calG}\left|\frac{1}{n}\sumi\left( g_1(U_i)-\E[ g(X_{12})]\right)\right|\right]\lesssim \sqrt{\frac{M^2v\log A }{n}}+\frac{vM\log A}{n}.
\end{align*}
This provides the bound for the first term of (\ref{eq:Hoeffding_decomposition}) in the non-degenerate case.\\

	 Now, define the classes of functions with support $(z,u_1,u_2)\in\supp(K)\times [0,1]^2$ by
	\begin{align*}
	&\calH_{2n}=\left\{(z,u_1,u_2)\mapsto x+zh_n : x\in \calX\right\},\\
	&\calH_{3}=\left\{(z,u_1,u_2)\mapsto u_1 \right\},\quad \calH_{3}'=\left\{(z,u_1,u_2)\mapsto u_2 \right\},\\
	&\calH_{4n}=\left\{(z,u_1,u_2)\mapsto f_{Y_{12}|U_1,U_2}(x+zh_n\mid u_1,u_2)  : x\in \calX\right\},\\
	&\overline \calH_{4n}=\left\{(u_1,u_2)\mapsto \int K(z) f_{Y_{12}|U_1,U_2}(x+zh_n| u_1,u_2) dz : x\in \calX\right\}.
	\end{align*} 
	In addition, define for each $(g_1,g_2,g_2')\in \calH_{2n}\times \calH_{3}\times \calH_{3}'$ the transformation $\phi(g_1,g_2,g_3)(z,u_1,u_2)=f_{Y_{12}|U_1,U_2}(g_1|g_2,g_2')$.
	Under Assumption \ref{a:rates}, using the same argument as in the previous paragraph, we establish that for some universal constants $A,v\ge 1$,
\begin{align*}
\sup_{Q}N(\overline \calH_{4n} ,\|\cdot\|_{Q,2},\epsilon M) \lesssim (A/\epsilon)^{5v}=O(1).
\end{align*}
Coming back to the second term of (\ref{eq:Hoeffding_decomposition}), it consists of a centered, completely degenerate $U$-process (see Section 3.5 in \cite{delaPenaGine1999} for definition). By the above entropy calculations along with Corollary 7(i) in \cite{Kato2017lecture}, we have $\sup_{Q}N(\overline \calH_{4n}-2\overline \calG_{4n} ,\|\cdot\|_{Q,2},\epsilon 2 M)=O(1)$ for all $\epsilon\in(0,1]$. An upper bound can be subsequently obtained by applying Corollary 5.3 in \cite{ChenKato2019b},
	\begin{align*}
	\E\left[\sup_{g\in\calG}\left|\frac{2}{n(n-1)}\sum_{1\le i<j\le n} W_{ij}(g)\right|\right]\lesssim \frac{M  v\log A}{n}+\frac{M v^2\log^2  A}{n^{3/2}}.
	\end{align*}

	\par	
	{  For the third term of (\ref{eq:Hoeffding_decomposition}), note that conditional on $(U_i)_{i=1}^n$, $(h_n \cdot R_{ij}(g))_{(i,j)\in I_n}$ are centered, independent but not necessarily identically distributed random processes. Thus an application of Lemma B.2 in \cite{cattaneo2022} with similar entropy calculations as above yields
	\begin{align*}
	&\frac{1}{h_n}\cdot\E\left[\sup_{g\in\calG}\left|\frac{2}{n(n-1)}\sum_\ijI h_n R_{ij}(g)\right|\mid (U_i)_{i=1}^n\right]\\
	\lesssim& \sigma_{U}(g)\sqrt{\frac{v\log(A\|K\|_\infty/\sigma_{U}(g))}{n^2 h_n^2}}
	+\frac{v\log(A\|K\|_\infty/\sigma_{U}(g))\|K\|_\infty}{n^2h_n}.
	\end{align*}
	where $\sigma_{U}^2(g):=\sup_{g\in\calG_n}{n \choose 2}^{-1}\sum_{1\le i<j\le n}\E[(h_n R_{ij}(g))^2\mid (U_i)_{i=1}^n]=O(h_n)$ almost surely following Assumption \ref{a:rates} and can be taken to be strictly positive following the definition of $\sigma$ in Lemma B.1 in \cite{cattaneo2022}.} By Jensen's inequality and Fubini's theorem, we have
	\begin{align*}
	\E\left[\sup_{g\in\calG}\left|\frac{2}{n(n-1)}\sum_\ijI R_{ij}(g)\right|\right]\lesssim \sqrt{\frac{v\log(A\|K\|_\infty/h_n)}{n^2 h_n}}
	+\frac{v\log(A\|K\|_\infty/h_n)\|K\|_\infty}{n^2h_n}.
	\end{align*}
	This shows the convergence rate of the stochastic component. The convergence rates for the two cases then follow from combining the bounds.
\end{proof}
\subsection{Proof of Theorem \ref{thm:asymptotic_dist}}
\begin{proof}\label{proof:asymptotic_dist}
	In this proof, we write $L_i$, $R_{ij}$ and $W_{ij}$ for  $L_i(g)$, $R_{ij}(g)$ and $W_{ij}(g)$ in Equation (\ref{eq:Hoeffding_decomposition}) where $g$ corresponds to the design point of interest $x$. Also let us denote $\kappa_n=nh_n$ and recall
	\begin{align*}
	\Omega_1(x)=\Var\left(f_{X_{12}|U_1}(x|U_1)\right),\qquad\Omega_2(x)=f(x)\int K^2(u) du.
	\end{align*}
	The first part of this proof follows a standard argument for empirical likelihood. For $\hat{\lambda}:=\text{argsup}_{\lambda}\sumi \log(1+\lambda V_i(\theta))$, the first-order condition and the algebraic fact that $(1+x)^{-1}=1-x+x^2(1+x)^{-1}$ yield
	\begin{align*}
	0=\frac{1}{n}\sumi \frac{V_i(\theta)}{1+\hat\lambda V_i(\theta)}=\frac{1}{n}\sumi V_i(\theta)-\frac{1}{n}\sumi V_i(\theta)^2\hat\lambda + \frac{1}{n}\sumi \frac{V_i(\theta)^3\hat \lambda^2}{1+\hat\lambda V_i(\theta)}.
	\end{align*}
	 Lemmas \ref{lem:third_moment}, \ref{lem:lambda_conv}, and \ref{lem:V2} together imply that under both non-degeneracy and degeneracy, it holds that
	\begin{align*}
	\hat \lambda=\frac{\sumi V_i(\theta)}{\sumi V_i(\theta)^2}+o_P(n^{-1/2}).
	\end{align*}
	Thus a Taylor expansion of $\log(1+x )$ at $1$ gives
	\begin{align*}
	2\sumi \log(1+\hat \lambda V_i(\theta) )=2\sumi\left(\hat \lambda V_i(\theta)-\frac{1}{2}\left\{\hat \lambda V_i(\theta)\right\}^2\right)+\op(1)=\frac{\left\{\frac{1}{n}\sumi V_i(\theta)\right\}^2}{\frac{1}{n}\sumi V_i(\theta)^2} +\op(1).
	\end{align*}
	Note that $\Var(\hat\theta)=4\Omega_1(x)/n+{ 2}\Omega_2(x)/n^2h_n+O(n^{-2})+O( h_n^2/n)$ following \citet[pp. 12]{GrahamNiuPowell2019}. Let $\sigma^2=4\Omega_1(x)+{ 2}\kappa_n^{-1}\Omega_2(x)$. Notice that $f(x)>0$ implies $\Omega_2>0$ and hence $\sigma^2>0$.
	It will be shown
	\begin{align*}
	\frac{1}{\sqrt{n} \sigma }\sumi V_i(\theta) \stackrel{d}{\to} N(0,1).
	\end{align*}
	To see this, note that
	\begin{align*}
	\frac{1}{n  }\sumi V_i(\theta)=\frac{1}{n  }\sumi\left(n S(\theta)-(n-1)S^{(i)}(\theta)\right)=S(\theta)
	\end{align*}
	using (\ref{eq:leave-one-out_trick}).
	Now, following the decomposition of Equation (\ref{eq:Hoeffding_decomposition}), we have
	\begin{align*}
	 S(\theta)
	=&\frac{1}{n}\sumi L_i +\frac{2}{n(n-1)}\sum_\ijI \left(W_{ij} +R_{ij}\right) + \text{Bias}(\hat \theta).
	\end{align*}
	Notice that the $L_i$, $W_{ij}$ and $R_{ij}$ terms above correspond to $T_2$, $T_3$ and $T_1$ in \citet[pp. 9]{GrahamNiuPowell2019} respectively. { Here the bias has an order of $h_n^2$. Thus, following the variance calculation in Section 3 in \cite{GrahamNiuPowell2019}, under 
nondegeneracy
\begin{align*}
 S(\theta)
	=&\frac{(1+o_P(1))}{n}\sumi L_i + \text{Bias}(\hat \theta)= O_P\left(\frac{1}{\sqrt{n}} \right)+ O\left(h_n^2\right)=O_P\left(\frac{1}{\sqrt{n}} \right),
\end{align*}	
and under degeneracy
	\begin{align*}
	 S(\theta)
	=&\frac{2}{n(n-1)}\sum_\ijI \left(W_{ij} +R_{ij}\right) + \text{Bias}(\hat \theta)=O_P\left(\frac{1}{\sqrt{n^2 h_n}} \right)+ O\left(h_n^2\right)=O_P\left(\frac{1}{\sqrt{n^2h_n}} \right).
	\end{align*}
	That is, the bias is negligible in both scenarios as $nh_n^{5/2}\to 0$.}
	Then, it follows from Section 4 in \citet[pp. 15 and 24]{GrahamNiuPowell2019} that
	\begin{align*}
	\frac{1}{\sqrt{n}\sigma}\sumi V_i(\theta) \stackrel{d}{\to} N(0,1).
	\end{align*}
	To investigate the behavior of $n^{-1}\sigma^{-2}\sumi V_i^2(\theta)$, note that following Lemma \ref{lem:V2}, under non-degeneracy, we have
	\begin{align*}
	\frac{1}{n\sigma^2}\sumi V_i(\theta)^2= \frac{4\Omega_1(x)}{\sigma^2}+o_P(1)\to 1,
	\end{align*}
	and under degeneracy,
	\begin{align*}
	\frac{1}{n\sigma^2}\sumi V_i(\theta)^2=\frac{1}{\sigma^2} { 4} \kappa_n^{-1 }\Omega_2(x)+o_P(1)\not\to 1.
	\end{align*}
	We now consider modified JEL. Notice that
	\begin{align*}
	\frac{1}{n\sigma^2}\sumi V_i^m(\theta)^2=\sigma^{-2}\left\{4\Omega_1(x)+ { 4}\kappa_n^{-1 }\Omega_2(x)\right\}+o_P(1),
	\end{align*}
	following a similar argument as above as well as the consistency of $\hat \theta$. It now suffices to show that
		\begin{align*}
	\frac{1}{\sqrt{n}\sigma}\sumi V_i^m(\theta)\stackrel{d}{\to} N\left(0, \lim_{n\to \infty}\sigma^{-2}\left\{4\Omega_1(x)+ { 4}\kappa_n^{-1 }\Omega_2(x)\right\}\right).
	\end{align*}
	As $\sumi V_i(\hat\theta)=0$, it holds that
	$
n^{-1/2}\sumi V_i^m(\theta)=\Gamma\Gamma_m^{-1} n^{-1/2}\sumi V_i(\theta)
	$.
	We have from above, $n^{-1/2}\sigma^{-1}\sumi V_i(\theta)\stackrel{d}{\to } N(0,1)$. Now, notice that
	\begin{align*}
\Gamma\Gamma_m^{-1}=\left\{\frac{n^{-1}\sigma^{-2}\sumi V_i(\hat\theta)^2}{n^{-1}\sigma^{-2}\sumi V_i(\hat\theta)^2-n^{-1}\sigma^{-2}\sum_{i=1}^{n-1} \sum_{j= i+1}^nQ_{ij}^2}\right\}^{1/2}.
	\end{align*}
	By consistency of $\hat \theta$ and the convergence of $n^{-1}\sigma^{-2}\sumi V_i(\theta)^2$, we have $n^{-1}\sigma^{-2}\sumi V_i(\hat\theta)^2\stackrel{P}{\to } \sigma^{-2}\{4\Omega_1(x) + { 4}\kappa_n^{-1 }\Omega_2(x)\}+o_P(1)$.
	We now verify that
	\begin{align*}
	\frac{1}{n\sigma^{2}}\sum_{i=1}^{n-1} \sum_{j= i+1}^nQ_{ij}^2=  \frac{ { 2}  }{\sigma^2}\kappa_n^{-1}\Omega_2(x)+o_P(1).
	\end{align*}
	To see this, observe that
	\begin{align*}
	(n-1)S^{(i)}(\theta)=&\frac{2}{(n-2)}\Bigg(\sum_{k\ne i} L_k +\sum_{k=1}^{n-1}\sum_{l=i+1}^n (W_{kl}+R_{kl})-\sum_{l= i+1}^n(W_{il}+R_{il})-\sum_{k =1}^{i-1} (W_{ki}+R_{ki})\Bigg)\\
	&+\text{Bias}(\hat \theta^{(i)}),\\
	(n-2)S^{(i,j)}(\theta) =&\frac{2}{(n-3)}\Bigg(\sum_{k\ne i,j} L_k+\sum_{k=1}^{n-1}\sum_{l=i+1}^n (W_{kl}+R_{kl})- \sum_{l=j+1}^n (W_{jl}+R_{jl}) -\sum_{k=1}^{j-1}(W_{kj}+R_{kj})\\
	 &\qquad\qquad\qquad -\sum_{l = i+1}^{n}(W_{il}+R_{il})-\sum_{k=1}^{i-1} (W_{ki}+R_{ki})  +(W_{ij}+R_{ij})\Bigg)+\text{Bias}(\hat \theta^{(i,j)}).
	\end{align*}
	Observe that $\text{Bias}(\hat \theta)=\text{Bias}(\hat \theta^{(i)})=\text{Bias}(\hat \theta^{(j)})=\text{Bias}(\hat \theta^{(i,j)})$ as they are estimated using the same bandwidth $h_n$.
	Hence, as the $W_{ij}$ terms are of smaller order asymptotically, we have
	\begin{align*}
	Q_{ij}=&\frac{1}{n-1}\left(\frac{2}{(n-1)(n-2)}\sum_{k=1}^{n-1}\sum_{l=k+1}^nR_{kl}-\frac{2}{n-2}\left(\sum_{l> i}R_{il}+\sum_{k< i}R_{ki}+\sum_{l> j}R_{jl}+\sum_{k< j}R_{kj}\right)
	 +2R_{ij }\right) \\
	 &\times (1+o_P(1)).
	\end{align*}
	Therefore, by the conditional independence of $(R_{ij})_{(i,j)\in\Ink}$, the WLLN for triangular arrays,  and the law of total variance,
	\begin{align*}
		\frac{1}{n\sigma^2}\sum_{i=1}^{n-1} \sum_{j= i+1}^nQ_{ij}^2	 =&\frac{4}{n (n-1)^2\sigma^2}\sum_{i=1}^{n-1}\sum_{j=i+1}^nR_{ij }^2 +o_P(1)\\
		  =&{ \frac{2}{n\sigma^2}\Var(R_{12}) +o_P(1)=\frac{2}{\sigma^2}\kappa_n^{-1}\Omega_2(x)+o_P(1).}
	\end{align*}
	This completes the proof.
\end{proof}

\subsection{Proof of Theorem \ref{thm:asymptotic_dist_incomplete}}
\begin{proof}\label{proof:asymptotic_dist_incomplete}
	In this proof, let us abbreviate $X_{\Ink}=(X_{ij})_{ \ijI}$. We will drop the $n$ subscript of $p$ with the understanding that it can potentially depend on $n$.
	
	 Define $\gamma=\lim_{n\to\infty}p/(1-p)$. 
	Denote $N=pn(n-1)/2$, $N_1=p(n-1)(n-2)/2$ and $N_2=(n-2)(n-3)/2$. First notice that $\hat p/p\stackrel{P}{\to} 1$ at the rate of $\{n(n-1)\}^{-1/2}$. Hence $\hat N/N\stackrel{P}{\to}1$, $\hat N_1/N_1\stackrel{P}{\to}1$ and $\hat N_2/N_2\stackrel{P}{\to}1$ at the same fast rate.
	Thus if we define
	\begin{align*}
	&\tilde S(\theta)=\frac{1}{ N}\sum_\ijI \frac{1}{h_n}Z_{ij }K\left(\frac{x-\Xij}{h_n}\right)-\theta,\quad \tilde S^{(k)}(\theta)=\frac{1}{ N_1}\sum_{\ijI^{(k)}}\frac{1}{h_n}Z_{ij}K\left(\frac{x-\Xij}{h_n}\right)-\theta,\\
&	\tilde S^{(k,l)}(\theta)=\frac{1}{ N_2}\sum_{\ijI^{(k,l)}}\frac{1}{h_n}Z_{ij} K\left(\frac{x-\Xij}{h_n}\right)-\theta,
	\end{align*}
	then we have
$
	\tilde S(\theta)=\hat S(\theta)+O_P(n^{-1}),
$
	and similar results hold uniformly for $\hat S^{(k)}(\theta)$ and $\hat S^{(k,l)}(\theta)$ over all $k,l$.
Denote 
\begin{align*}
&\tilde V_i(\theta):=n\tilde S(\theta)-(n-1) \tilde S^{(i)}(\theta),\\
&\tilde V_i^m(\theta):=\tilde V_i(\hat \theta_{\text{inc}})-\hat \Gamma\hat \Gamma_m^{-1}\{\tilde V_i(\hat \theta_{\text{inc}})-\tilde V_i(\theta)\}.
\end{align*}
	Using the relationship between $\hat S$ and $\tilde S$ from above, and the facts
	$
	n^{-1}\sumi \tilde S^{(i)}(\theta)=\tilde S(\theta)
	$ and $
	n^{-1}\sumi \hat S^{(i)}(\theta)=\hat S(\theta)
	$ that follow Equation (\ref{eq:leave-one-out_trick}) with $Z_{ij}K_{ij,n}$ in place of $K_{ij,n}$, we have
	\begin{align*}
	\frac{1}{\sqrt{n}}\sumi \tilde V_i(\theta)=\sqrt{n} \tilde S(\theta)=\sqrt{n}\hat S(\theta)+O_p(n^{-1/2})=	\frac{1}{\sqrt{n}}\sumi \hat V_i(\theta)+O_p(n^{-1/2}).
	\end{align*}
	\\
	\par	
	Now, the same argument from the first part of the Proof of Theorem \ref{thm:asymptotic_dist} holds for $\tilde V_i^m(\theta)$ in place of $V_i(\theta)$ and thus we have
	\begin{align*}
	2\sumi \log(1+\hat \lambda \tilde V_i^m(\theta) )=\frac{\left\{\frac{1}{n}\sumi \tilde V_i^m(\theta)\right\}^2}{\frac{1}{n}\sumi \tilde V_i^m(\theta)^2} +\op(1).
	\end{align*}
	Denote 
	$
\tilde \sigma^2=\sigma^2+\gamma^{-1}{ 2}\kappa_n^{-1}\Omega_2(x)
$,
	where $\sigma^2=4\Omega_1(x)+ { 2}\kappa_n^{-1 }\Omega_2(x)$ and $\kappa_n=nh_n$ are the same as in the proof of Theorem \ref{thm:asymptotic_dist}.
	To this end, it suffices to verify that for $\omega^2:=\lim_{n\to \infty}\tilde \sigma^{-2}\{4\Omega_1(x) +(1+\gamma^{-1}){ 2}\kappa_n^{-1}\Omega_2(x)\}$, it holds that
	\begin{align*}
	\frac{1}{\sqrt{n}\tilde\sigma}\sumi \tilde V_i^m(\theta) \stackrel{d}{\to} N(0,\omega^2),\quad 
	\frac{1}{n\tilde \sigma^2}\sumi \tilde V_i^m(\theta)^2 \stackrel{P}{\to} \omega^2.
	\end{align*}
	Let us start with $\tilde V_i$, notice that we can decompose	
	\begin{align*}
	\frac{1}{\sqrt{n}\tilde \sigma}\sumi \tilde V_i(\theta)
	=&\sqrt{n}\tilde \sigma^{-1}\tilde S(\theta)=\frac{\sqrt{n}}{ N\tilde \sigma}\sum_\ijI \frac{1}{h_n}Z_{ij }K\left(\frac{x-\Xij}{h_n}\right)-\theta+\frac{\sqrt{n}}{N\tilde \sigma}\sum_\ijI(Z_{ij}-p)K_{ij,n}\\
	=&\sqrt{n}\tilde \sigma^{-1}\left\{S(\theta) + T(\theta)\right\},
	\end{align*}
	where $S(\theta)$ is the same as in the complete dyadic data case and 
	$
	T(\theta)=N^{-1}\sum_\ijI(Z_{ij}-p)K_{ij,n}
	$.
	From the proof of Theorem \ref{thm:asymptotic_dist}, we have $\sqrt{n}\sigma^{-1}S(\theta)\stackrel{d}{\to} N(0,1)$. Now, conditional on $ X_{\Ink}$, $\sqrt{n}\tilde \sigma^{-1} T(\theta)$ consists of a sum of independent random variable with mean zero and conditional variance
	\begin{align*}
	{ \Var\left(T(\theta) \mid X_{\Ink} \right)=\frac{4(1-p)}{pn^2(n-1)^2}\sum_{\ijI}K_{ij,n}^2.}
	\end{align*}
By the law of total variance, as $E[T(\theta)\mid X_{\Ink} ]=0$, we have
\begin{align*}
\Var(T(\theta))=\frac{(1-p)}{p}\frac{{ 2}}{n^2 h_n }\Omega_2(x)(1+o(1)).
\end{align*}
In addition, one can verify the conditional Lyapunov's condition
\begin{align}
\frac{1}{\{\Var(T(\theta))\}^{3/2}}\E\left[\left|T(\theta)\right|^3\mid X_{\Ink}\right]\lesssim
\frac{1}{\{\Var(T(\theta))\}^{3/2}}\frac{p}{p^3 n^6}\sum_{(i,j)\in \Ink} K_{ij,n}^3.\label{eq:lyapunov}
\end{align}
Notice that as Assumption \ref{a:JEL_incomplete} implies $n^2 h_n p\to \infty$, the right hand side of (\ref{eq:lyapunov}) is of order $o_P(1)$ since
\begin{align*}
\E\left[\left|\frac{1}{\{\Var(T(\theta))\}^{3/2}}\frac{p}{p^3 n^6}\sum_{(i,j)\in \Ink} K_{ij,n}^3\right|\right]\lesssim\frac{1}{n  h_n^{1/2} p^{1/2}}=o(1).
\end{align*}
Thus an application of Lyapunov's CLT yields that with probability $1-o(1)$, it holds that
\begin{align*}
{ \Var( T(\theta) \mid X_{\Ink} )^{-1/2}T(\theta)\stackrel{d|X_{I_n}}{\to} N\left(0,1
\right)
,}
\end{align*} 
where the convergence in distribution is with respect to the law of $(Z_{ij})_{(i,j)\in\Ink}$. The variance calculation above and conditional convergence in distribution then imply unconditional convergence in distribution $\{\Var(T(\theta))\}^{-1/2}T(\theta)\stackrel{d}{\to}N(0,1)$. 
Hence, an application of Theorem 2 in \cite{chen2007asymptotic} together with Slutsky's Theorem yields that
	\begin{align*}
	\sqrt{n}\tilde \sigma^{-1}\tilde S(\theta)=\sqrt{n}\tilde \sigma^{-1}\left\{S(\theta) + T(\theta)\right\}\stackrel{d}{\to}N\left(0,\lim_{n\to \infty}\{\tilde \sigma^{-2}\sigma^2+\tilde \sigma^{-2}\gamma^{-1}{ 2}\kappa_n^{-1}\Omega_2(x) \}\right)=N\left(0,1\right).
	\end{align*}
	Now, we claim
	\begin{align*}
	\frac{1}{n\tilde \sigma^2}\sumi \tilde V_i(\theta)^2= \tilde \sigma^{-2}\{4\Omega_1(x) + { 4}\kappa_n^{-1 }\Omega_2(x)\}+\tilde \sigma^{-2}\gamma^{-1}{ 2}\kappa_n^{-1}\Omega_2(x)+o_P(1).
	\end{align*}
	Define $T^{(i)}(\theta)=N_1^{-1}\sum_{(k,l)\in \Ink^{(i)}}(Z_{kl}-p)K_{kl,n}$, the leave-one-out version of $T$.
	Note that $\tilde V_i(\theta)=\tilde S(\theta)-(n-1)(\tilde S^{(i)}(\theta)-\tilde S(\theta))$. Then
	\begin{align*}
	\frac{1}{n\tilde \sigma^2}\sumi \tilde V_i^2(\theta)=&\tilde \sigma^{-2}\left(\tilde S^2(\theta) - 2(n-1) \tilde  S(\theta)\frac{1}{n} \sumi \{\tilde  S^{(i)}(\theta)-\tilde  S(\theta)\} +(n-1)^2\frac{1}{n}\sumi\{\tilde  S^{(i)}(\theta)-\tilde  S(\theta)\}^2\right)\\
	=&\tilde \sigma^{-2}\left(M_1-2M_2+M_3\right).
	\end{align*}
	Observe that $M_1=o_P(1)$ following the uniform convergence rate from Theorem \ref{thm:unif_rate} and $M_2=0$ using Equation (\ref{eq:leave-one-out_trick}). Similar to the calculation of term $T_3$ in the Proof of Theorem \ref{thm:asymptotic_dist},
	\begin{align*}
	\tilde \sigma^{-2}M_3=&\frac{(n-1)^2}{n\tilde \sigma^2}\sumi\{\tilde  S^{(i)}(\theta)-\tilde  S(\theta)\}^2\\
	=&\frac{(n-1)^2}{n\tilde \sigma^2}\sumi\left\{\tilde  S^{(i)}(\theta)-\frac{1}{n}\sumi \tilde  S^{(i)}(\theta)\right\}^2\\
	=&\frac{(n-1)^2}{n^2\tilde \sigma^2}\sum_{i=1}^{n-1}\sum_{j=i+1}^n\left\{\tilde  S^{(i)}(\theta)-\tilde  S^{(j)}(\theta)\right\}^2\\
	=&\frac{(n-1)^2}{n^2\tilde \sigma^2}\sum_{i=1}^{n-1}\sum_{j=i+1}^n\left\{ \left(S^{(i)}(\theta)+T^{(i)}(\theta)\right)-\left(S^{(j)}(\theta)+T^{(j)}(\theta)\right)\right\}^2\\
	=&\frac{(n-1)^2}{n^2\tilde \sigma^2}\sum_{i=1}^{n-1}\sum_{j=i+1}^n\left\{ S^{(i)}(\theta)-S^{(j)}(\theta)\right\}^2+\frac{(n-1)^2}{n^2\tilde \sigma^2}\sum_{i=1}^{n-1}\sum_{j=i+1}^n\left\{ T^{(i)}(\theta)-T^{(j)}(\theta)\right\}^2\\
	&+\frac{(n-1)^2}{n^2\tilde \sigma^2}\sum_{i=1}^{n-1}\sum_{j=i+1}^n2\left\{ S^{(i)}(\theta)-S^{(j)}(\theta)\right\}\left\{ T^{(i)}(\theta)-T^{(j)}(\theta)\right\}\\
	=&\tilde \sigma^{-2}\{4\Omega_1(x) + { 4}\kappa_n^{-1 }\Omega_2(x)\}+ { 2}\tilde\sigma^{-2}\gamma^{-1}\kappa_n^{-1}\Omega_2(x)+o_P(1)=\omega^2+o_P(1),
	\end{align*}
	where the convergence stated in the last line is a direct result of the WLLN for triangular arrays and the conditional independence (given $X_{\Ink}$) between $S^{(i)}(\theta)$ and $T^{(i)}(\theta)$, as well as the convergence of the term $T_3$ in the Proof of Theorem \ref{thm:asymptotic_dist}.
The same proof strategy yields that $	n^{-1}\tilde \sigma^{-2}\sumi \tilde V_i^m(\theta)^2 \stackrel{P}{\to} \omega^2$. \\
\par
It remains to show that 
	\begin{align*}
\frac{1}{	\sqrt{n}\tilde\sigma}\sumi \tilde V_i^m(\theta) \stackrel{d}{\to} N(0,\omega^2).
	\end{align*}
	 Note that $\sumi \hat V_i(\hat\theta_{\text{inc}})=0$ and hence
	\begin{align*}
	\frac{1}{\sqrt{n}\tilde\sigma}\sumi \hat V_i^m(\theta)=\left\{\frac{n^{-1}\tilde\sigma^{-2}\sumi \hat V_i(\hat\theta_{\text{inc}})^2}{n^{-1}\tilde\sigma^{-2}\sumi \hat V_i(\hat\theta_{\text{inc}})^2-n^{-1}\tilde\sigma^{-2}\sum_{i=1}^{n-1} \sum_{j= i+1}^n\hat Q_{ij}^2}\right\}^{1/2}\cdot
	\frac{1}{\sqrt{n}\tilde\sigma}\sumi \hat V_i(\theta).
	\end{align*}
	We have from above, $n^{-1/2}\tilde \sigma^{-1}\sumi\hat V_i(\theta)\stackrel{d}{\to } N(0,1)$ following $n^{-1/2}\tilde \sigma^{-1}\sumi\tilde V_i(\theta)\stackrel{d}{\to } N(0,1)$. 
	By consistency of $\hat \theta_{\text{inc}}$, we have $n^{-1}\tilde \sigma^{-2}\sumi \hat V_i(\hat\theta_{\text{inc}})^2\stackrel{P}{\to } \omega^2$. It suffices to show that the denominator in the square root converges to unity asymptotically. Some calculations yield that
	\begin{align*}
	\hat Q_{ij}=&\frac{1}{p(n-1)}\Bigg\{\frac{2}{(n-1)(n-2)}\sum_{k=1}^{n-1}\sum_{l> k}Z_{kl}R_{kl}-\frac{2}{n-2}\left(\sum_{l\ne  i}Z_{il}R_{il}+\sum_{l\ne j}Z_{jl}R_{jl}\right)\\
	&\qquad\qquad+2Z_{ij}R_{ij }\Bigg\}(1 +o_P(1)).
	\end{align*}
	Therefore, by the conditional independence (given $(U_i)_{i=1}^n$) of $(Z_{ij}R_{ij})_{(i,j)\in\Ink}$, using the WLLN for triangular arrays, the law of total variance, and following the same arguments as in the proof of Theorem \ref{thm:asymptotic_dist}, we have $
n^{-1}\tilde\sigma^{-2}\sum_{i=1}^{n-1} \sum_{j= i+1}^{n}\hat Q_{ij}^2= \tilde \sigma^{-2}{ 2}\kappa_n^{-1}\Omega_2(x)+o_P(1)
$.
This concludes the proof for the modified JEL. 

The proof for the JEL follows from similar arguments and an observation of 
\begin{align*}
\tilde \sigma^2=4\Omega_1(x)+ \kappa_n^{-1 }\Omega_2(x)+\gamma^{-1}\kappa_n^{-1}\Omega_2(x)=
\begin{cases}
4\Omega_1(x)+o(1), &\text{ if $f$ is non-degenerate at $x$,}\\ 
\frac{{ 2}}{nh_n p}\Omega_2(x)(1+o(1)),& \text{ if $f$ is degenerate at $x$ and $p\to0$.}
\end{cases}
\end{align*}
\end{proof}

\newpage
\begin{table}
	\begin{center}
		\begin{tabular}{cccccc}
			\hline
			\hline
			\multicolumn{6}{c}{dyadic DGP  ($\beta=1$)}\\
			\hline
			\hline
			kernel: &\multicolumn{5}{c}{Epanechnikov}\\
			bandwidth: &\multicolumn{5}{c}{$h_n^{S}$}\\
			\hline
			method:&FG&lFG&JEL&mJK&mJEL\\
			\hline
			n=25& 0.863 & 0.986 & 0.985 & 0.919 &  0.933 \\
			n=50& 0.906 & 0.981 & 0.982 & 0.927 &  0.937 \\
			n=75& 0.927 & 0.984 & 0.985 & 0.941 & 0.947 \\
			\hline
			n=100& 0.929 & 0.983 & 0.984 &  0.941 & 0.948 \\
			n=150& 0.926 & 0.981 & 0.981 & 0.936 & 0.944 \\
			n=200& 0.939 & 0.979 & 0.980 &  0.942 & 0.947 \\
			\hline
			n=400& 0.936 & 0.970 & 0.971 & 0.938 & 0.942 \\
			n=600& 0.944 & 0.971 & 0.971 & 0.945 & 0.948 \\
			n=800& 0.945 & 0.971 & 0.971 & 0.946 & 0.948 \\
			\hline
		\end{tabular}
	\end{center}
	\caption{Coverage probabilities of $95$\% confidence intervals under non-degenerate DGP}
	\label{table:dyadic_cov_prob}
\end{table}

\begin{table}
	\begin{center}
		\begin{tabular}{cccccc}
			\hline
			\hline
			\multicolumn{6}{c}{i.i.d. DGP ($\beta=0$)}\\
			\hline
			\hline
			kernel: &\multicolumn{5}{c}{Epanechnikov}\\
			bandwidth: &\multicolumn{5}{c}{$h_n^{S}$}\\
			\hline
			method:&FG&lFG&JEL&mJK&mJEL\\
			\hline
			n=25& 0.851 & 0.983 & 0.987 & 0.920 & 0.937 \\
			n=50& 0.905 & 0.991 & 0.991 & 0.936 & 0.952 \\
			n=75& 0.927 & 0.994 & 0.994 & 0.944 & 0.956 \\
			\hline
			n=100& 0.934  & 0.993 & 0.994 & 0.946 & 0.957 \\
			n=150& 0.940 & 0.994 & 0.994 & 0.950 & 0.959 \\
			n=200& 0.945  & 0.993 & 0.994 & 0.951 & 0.961 \\
			\hline
			n=400& 0.943 & 0.992 & 0.992 & 0.944 & 0.954 \\
			n=600& 0.939 & 0.992 & 0.992 & 0.941 & 0.953 \\
			n=800& 0.944 & 0.995 & 0.995 & 0.945 & 0.955 \\
			\hline
		\end{tabular}
	\end{center}
	\caption{Coverage probabilities of $95$\% confidence intervals under degenerate DGP }
	\label{table:iid_cov_prob}
\end{table}

{ \small
	
	\begin{table}
		\begin{center}
		\begin{tabular}{ccc}
			\hline
			\hline
			\multicolumn{3}{c}{incomplete dyadic DGP ($p_n=0.5$ and $\beta=1$)}\\
			\hline
			\hline
			kernel: &\multicolumn{2}{c}{Epanechnikov}\\
			bandwidth: &\multicolumn{2}{c}{$h_n^{S, inc}$}\\
			\hline
			method:& JEL&mJEL\\
	\hline
			n=100&0.988&0.952\\
			n=150&0.986&0.950\\
			n=200&0.984&0.951\\
			\hline
			n=400&0.982&0.951\\
			n=600&0.984&0.952\\
			n=800&0.977&0.953\\
			\hline
		\end{tabular}
	\end{center}
		\caption{Coverage probabilities of $95$\% confidence intervals under incomplete dyadic DGP }
		\label{table:incomplete_dyadic}
	\end{table}
	
	\begin{table}
		\begin{center}
		\begin{tabular}{ccc}
			\hline
			\hline
			\multicolumn{3}{c}{incomplete dyadic DGP ($p_n=0.5$ and $\beta=0$)}\\
			\hline
			\hline
			kernel: &\multicolumn{2}{c}{Epanechnikov}\\
			bandwidth: &\multicolumn{2}{c}{$h_n^{S, inc}$}\\
			\hline
			method:& JEL&mJEL\\
			\hline
			n=100&0.993&0.956\\
			n=150&0.994&0.960\\
			n=200&0.994&0.957\\
			\hline
			n=400&0.997&0.960\\
			n=600&0.994&0.960\\
			n=800&0.994&0.958\\
			\hline
		\end{tabular}
	\end{center}
		\caption{Coverage probabilities of $95$\% confidence intervals under incomplete i.i.d. DGP }
		\label{table:incomplete_iid}
	\end{table}
	
}
\begin{figure}
	\caption{Distribution of summary statistics of taxi time across routes}
	\begin{minipage}{\linewidth}
		\subcaption{Mean}
		\includegraphics[width=\textwidth]{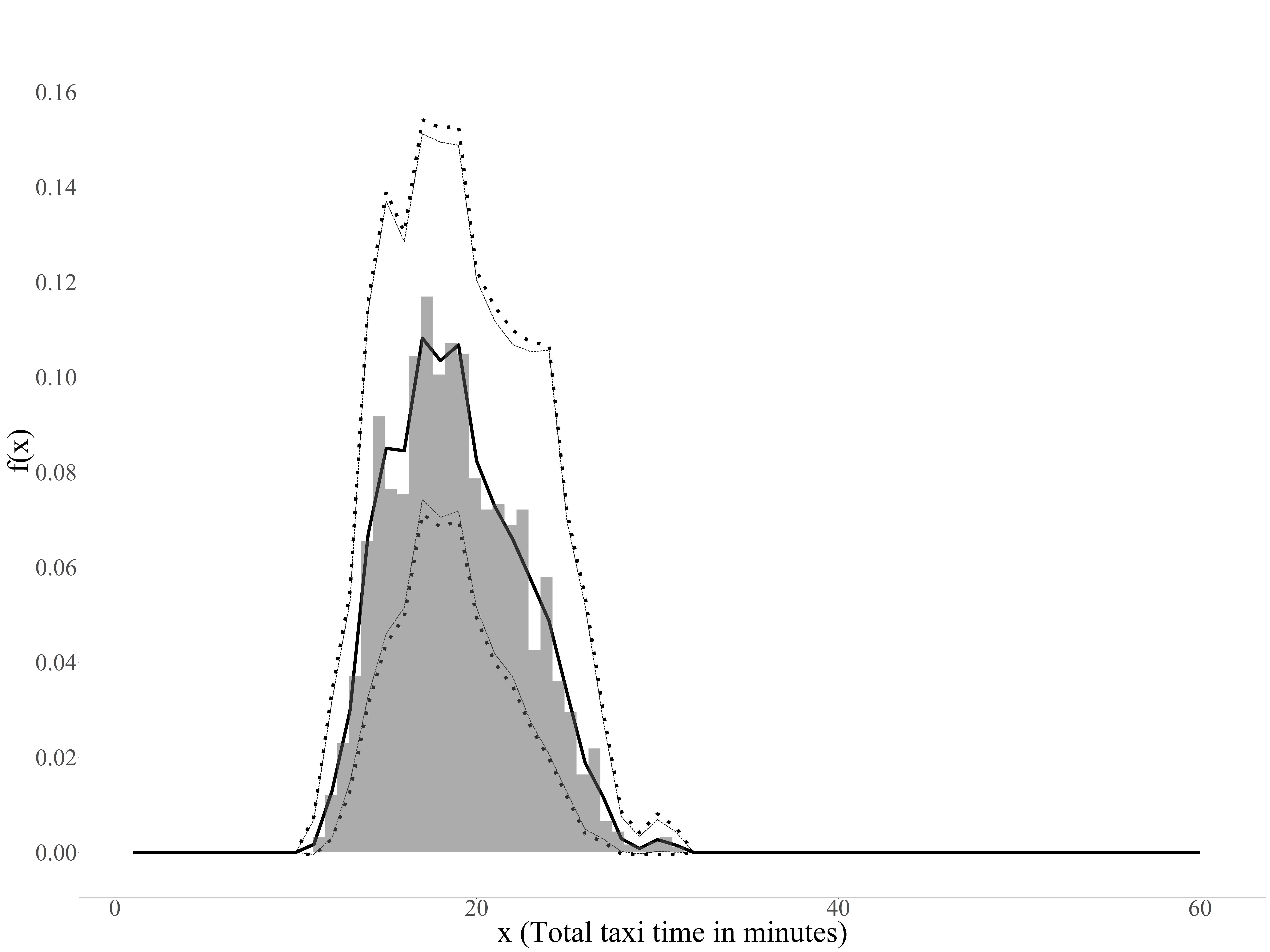}	
				\caption*{\footnotesize Point estimates for the KDE, evaluated at each whole minute, are represented by the solid line. Pointwise 95\% modified JEL confidence intervals are represented by the two dashed lines. Pointwise 95\% JEL confidence intervals are represented by the two dotted lines.}
	\end{minipage}
\end{figure}
\begin{figure}\ContinuedFloat
	\begin{minipage}{\linewidth}
	\subcaption{95th percentile}
	\includegraphics[width=\textwidth]{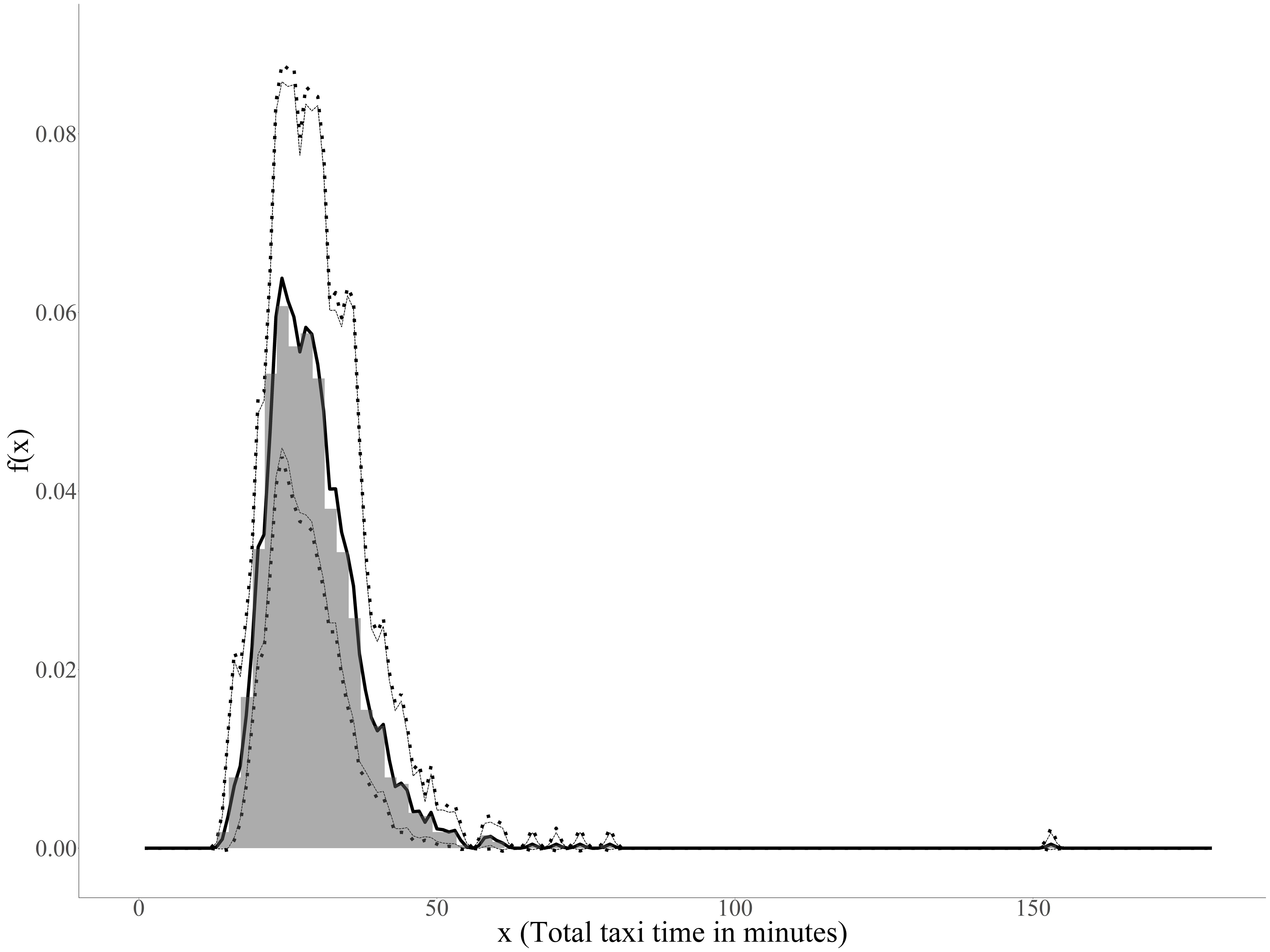}	
				\caption*{\footnotesize Point estimates for the KDE, evaluated at each whole minute, are represented by the solid line. Pointwise 95\% modified JEL confidence intervals are represented by the two dashed lines. Pointwise 95\% JEL confidence intervals are represented by the two dotted lines.}
\end{minipage}
\end{figure}
\begin{figure}\ContinuedFloat
\begin{minipage}{\linewidth}
	\subcaption{Maximum}
	\includegraphics[width=\textwidth]{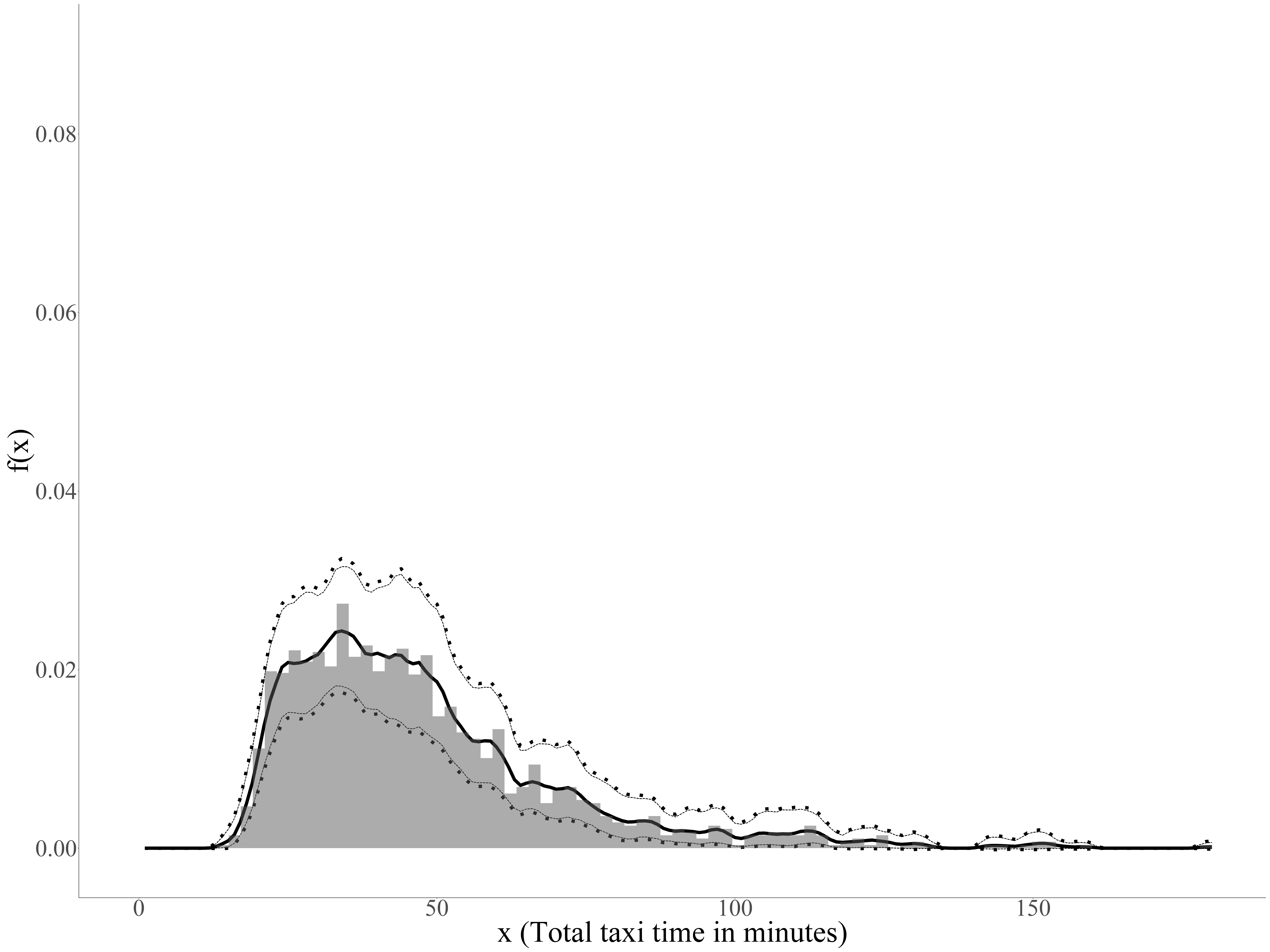}	
			\caption*{\footnotesize Point estimates for the KDE, evaluated at each whole minute, are represented by the solid line. Pointwise 95\% modified JEL confidence intervals are represented by the two dashed lines. Pointwise 95\% JEL confidence intervals are represented by the two dotted lines.}
\end{minipage}
\end{figure}
\bibliographystyle{ecta}
\bibliography{biblio}

\begin{thebibliography}{59}
\newcommand{\enquote}[1]{``#1''}
\expandafter\ifx\csname natexlab\endcsname\relax\def\natexlab#1{#1}\fi

\bibitem[\protect\citeauthoryear{Akritas and Van~Keilegom}{Akritas and
  Van~Keilegom}{2001}]{AkritasVan_Keilegom2001}
\textsc{Akritas, M.~G. and I.~Van~Keilegom} (2001): \enquote{Non-parametric
  estimation of the residual distribution,} \emph{Scandinavian Journal of
  Statistics}, 28, 549--567.

\bibitem[\protect\citeauthoryear{Aronow, Samii, and Assenova}{Aronow
  et~al.}{2015}]{AronowSamiiAssenova2015}
\textsc{Aronow, P.~M., C.~Samii, and V.~A. Assenova} (2015):
  \enquote{Cluster--robust variance estimation for dyadic data,}
  \emph{Political Analysis}, 23, 564--577.

\bibitem[\protect\citeauthoryear{Bickel, Chen, and Levina}{Bickel
  et~al.}{2011}]{BickelChenLevina2011}
\textsc{Bickel, P.~J., A.~Chen, and E.~Levina} (2011): \enquote{The method of
  moments and degree distributions for network models,} \emph{Annals of
  Statistics}, 39, 2280--2301.

\bibitem[\protect\citeauthoryear{Bravo, Escanciano, and Van~Keilegom}{Bravo
  et~al.}{2020}]{BravoEscancianoVan_Keilegom2020AoS}
\textsc{Bravo, F., J.~C. Escanciano, and I.~Van~Keilegom} (2020):
  \enquote{Two-step semiparametric empirical likelihood inference,}
  \emph{Annals of Statistics}, 48, 1--26.

\bibitem[\protect\citeauthoryear{Calonico, Cattaneo, and Farrell}{Calonico
  et~al.}{2018{\natexlab{a}}}]{calonico2018effect}
\textsc{Calonico, S., M.~D. Cattaneo, and M.~H. Farrell} (2018{\natexlab{a}}):
  \enquote{On the effect of bias estimation on coverage accuracy in
  nonparametric inference,} \emph{Journal of the American Statistical
  Association}, 113, 767--779.

\bibitem[\protect\citeauthoryear{Calonico, Cattaneo, and Farrell}{Calonico
  et~al.}{2018{\natexlab{b}}}]{CalonicoCattaneoFarrell2018}
---\hspace{-.1pt}---\hspace{-.1pt}--- (2018{\natexlab{b}}): \enquote{On the
  effect of bias estimation on coverage accuracy in nonparametric inference,}
  \emph{Journal of the American Statistical Association}, 113, 767--779.

\bibitem[\protect\citeauthoryear{Cameron and Miller}{Cameron and
  Miller}{2014}]{CameronMiller2014}
\textsc{Cameron, A.~C. and D.~L. Miller} (2014): \enquote{Robust inference for
  dyadic data,} \emph{Unpublished manuscript, University of California-Davis}.

\bibitem[\protect\citeauthoryear{Cattaneo, Feng, and Underwood}{Cattaneo
  et~al.}{2022}]{cattaneo2022}
\textsc{Cattaneo, M., Y.~Feng, and W.~Underwood} (2022): \enquote{Uniform
  inference for kernel density estimators with dyadic data,} \emph{preprint}.

\bibitem[\protect\citeauthoryear{Cattaneo, Crump, and Jansson}{Cattaneo
  et~al.}{2013}]{CattaneoCrumpJansson2013JASA}
\textsc{Cattaneo, M.~D., R.~K. Crump, and M.~Jansson} (2013):
  \enquote{Generalized jackknife estimators of weighted average derivatives,}
  \emph{Journal of the American Statistical Association}, 108, 1243--1256.

\bibitem[\protect\citeauthoryear{Cattaneo, Crump, and Jansson}{Cattaneo
  et~al.}{2014{\natexlab{a}}}]{CattaneoCrumpJansson2014bootstrap}
---\hspace{-.1pt}---\hspace{-.1pt}--- (2014{\natexlab{a}}):
  \enquote{Bootstrapping density-weighted average derivatives,}
  \emph{Econometric Theory}, 1135--1164.

\bibitem[\protect\citeauthoryear{Cattaneo, Crump, and Jansson}{Cattaneo
  et~al.}{2014{\natexlab{b}}}]{CattaneoCrumpJansson2014ET}
---\hspace{-.1pt}---\hspace{-.1pt}--- (2014{\natexlab{b}}): \enquote{Small
  bandwidth asymptotics for density-weighted average derivatives,}
  \emph{Econometric Theory}, 176--200.

\bibitem[\protect\citeauthoryear{Cattaneo, Jansson, and Ma}{Cattaneo
  et~al.}{2020{\natexlab{a}}}]{CattaneoMaJansson2020wp}
\textsc{Cattaneo, M.~D., M.~Jansson, and X.~Ma} (2020{\natexlab{a}}):
  \enquote{Local regression distribution estimators,} \emph{Working paper}.

\bibitem[\protect\citeauthoryear{Cattaneo, Jansson, and Ma}{Cattaneo
  et~al.}{2020{\natexlab{b}}}]{CattaneoMaJansson2020JASA}
---\hspace{-.1pt}---\hspace{-.1pt}--- (2020{\natexlab{b}}): \enquote{Simple
  local polynomial density estimators,} \emph{Journal of the American
  Statistical Association}, 115, 1449--1455.

\bibitem[\protect\citeauthoryear{Cattaneo, Jansson, and Newey}{Cattaneo
  et~al.}{2018{\natexlab{a}}}]{cattaneo2018alternative}
\textsc{Cattaneo, M.~D., M.~Jansson, and W.~K. Newey} (2018{\natexlab{a}}):
  \enquote{Alternative asymptotics and the partially linear model with many
  regressors,} \emph{Econometric Theory}, 34, 277--301.

\bibitem[\protect\citeauthoryear{Cattaneo, Jansson, and Newey}{Cattaneo
  et~al.}{2018{\natexlab{b}}}]{cattaneo2018inference}
---\hspace{-.1pt}---\hspace{-.1pt}--- (2018{\natexlab{b}}): \enquote{Inference
  in linear regression models with many covariates and heteroscedasticity,}
  \emph{Journal of the American Statistical Association}, 113, 1350--1361.

\bibitem[\protect\citeauthoryear{Chen and Rao}{Chen and
  Rao}{2007}]{chen2007asymptotic}
\textsc{Chen, J. and J.~Rao} (2007): \enquote{Asymptotic normality under
  two-phase sampling designs,} \emph{Statistica Sinica}, 1047--1064.

\bibitem[\protect\citeauthoryear{Chen and Tabri}{Chen and
  Tabri}{2020}]{ChenTabri2020}
\textsc{Chen, R. and R.~V. Tabri} (2020): \enquote{Jackknife empirical
  likelihood for inequality constraints on regular functionals,} \emph{Journal
  of Econometrics}.

\bibitem[\protect\citeauthoryear{Chen and Van~Keilegom}{Chen and
  Van~Keilegom}{2009}]{chen2009review}
\textsc{Chen, S.~X. and I.~Van~Keilegom} (2009): \enquote{A review on empirical
  likelihood methods for regression,} \emph{Test}, 18, 415--447.

\bibitem[\protect\citeauthoryear{Chen and Kato}{Chen and
  Kato}{2019}]{ChenKato2019b}
\textsc{Chen, X. and K.~Kato} (2019): \enquote{Jackknife multiplier bootstrap:
  finite sample approximations to the U-process supremum with applications,}
  \emph{Probability Theory and Related Fields}, 1--67.

\bibitem[\protect\citeauthoryear{Chernozhukov, Chetverikov, and
  Kato}{Chernozhukov et~al.}{2014}]{CCK2014AoS}
\textsc{Chernozhukov, V., D.~Chetverikov, and K.~Kato} (2014):
  \enquote{Gaussian approximation of suprema of empirical processes,}
  \emph{Annals of Statistics}, 42, 1564--1597.

\bibitem[\protect\citeauthoryear{Chiang, Kato, and Sasaki}{Chiang
  et~al.}{2020}]{ChiangKatoSasaki2020}
\textsc{Chiang, H.~D., K.~Kato, and Y.~Sasaki} (2020): \enquote{Inference for
  high-dimensional exchangeable arrays,} \emph{arXiv preprint
  arXiv:2009.05150}.

\bibitem[\protect\citeauthoryear{Davezies, D’Haultf{\oe}uille, and
  Guyonvarch}{Davezies et~al.}{2021}]{DDG2020}
\textsc{Davezies, L., X.~D’Haultf{\oe}uille, and Y.~Guyonvarch} (2021):
  \enquote{Empirical process results for exchangeable arrays,} \emph{Annals of
  Statistics}, 49, 845--862.

\bibitem[\protect\citeauthoryear{de~la Pe\~na and Gin\'e}{de~la Pe\~na and
  Gin\'e}{1999}]{delaPenaGine1999}
\textsc{de~la Pe\~na, V. and E.~Gin\'e} (1999): \emph{Decoupling: From
  Dependence to Independence}, Springer.

\bibitem[\protect\citeauthoryear{Dony and Mason}{Dony and
  Mason}{2008}]{DonyMason2008}
\textsc{Dony, J. and D.~M. Mason} (2008): \enquote{Uniform in bandwidth
  consistency of conditional $ U $-statistics,} \emph{Bernoulli}, 14,
  1108--1133.

\bibitem[\protect\citeauthoryear{Efron and Stein}{Efron and
  Stein}{1981}]{EfronStein1981}
\textsc{Efron, B. and C.~Stein} (1981): \enquote{The jackknife estimate of
  variance,} \emph{Annals of Statistics}, 9, 586--596.

\bibitem[\protect\citeauthoryear{Einmahl and Mason}{Einmahl and
  Mason}{2000}]{EinmahlMason2000}
\textsc{Einmahl, U. and D.~M. Mason} (2000): \enquote{An empirical process
  approach to the uniform consistency of kernel-type function estimators,}
  \emph{Journal of Theoretical Probability}, 13, 1--37.

\bibitem[\protect\citeauthoryear{Einmahl and Mason}{Einmahl and
  Mason}{2005}]{EinmahlMason2005}
---\hspace{-.1pt}---\hspace{-.1pt}--- (2005): \enquote{Uniform in bandwidth
  consistency of kernel-type function estimators,} \emph{Annals of Statistics},
  33, 1380--1403.

\bibitem[\protect\citeauthoryear{Escanciano and Jacho-Ch{\'a}vez}{Escanciano
  and Jacho-Ch{\'a}vez}{2012}]{EscancianoJacho-Chavez2012}
\textsc{Escanciano, J.~C. and D.~T. Jacho-Ch{\'a}vez} (2012):
  \enquote{$\sqrt{n}$-uniformly consistent density estimation in nonparametric
  regression models,} \emph{Journal of Econometrics}, 167, 305--316.

\bibitem[\protect\citeauthoryear{Fafchamps and Gubert}{Fafchamps and
  Gubert}{2007}]{FafchampsGubert2007}
\textsc{Fafchamps, M. and F.~Gubert} (2007): \enquote{The formation of risk
  sharing networks,} \emph{Journal of Development Economics}, 83, 326--350.

\bibitem[\protect\citeauthoryear{Frank and Snijders}{Frank and
  Snijders}{1994}]{FrankSnijders1994}
\textsc{Frank, O. and T.~Snijders} (1994): \enquote{Estimating the size of
  hidden populations using snowball sampling,} \emph{Journal of Official
  Statistics}, 10, 53--53.

\bibitem[\protect\citeauthoryear{Ghosal, Sen, and van~der Vaart}{Ghosal
  et~al.}{2000}]{GhosalSenvdV2000}
\textsc{Ghosal, S., A.~Sen, and A.~W. van~der Vaart} (2000): \enquote{Testing
  monotonicity of regression,} \emph{Annals of Statistics}, 1054--1082.

\bibitem[\protect\citeauthoryear{Gin{\'e} and Nickl}{Gin{\'e} and
  Nickl}{2016}]{GineNickl2016}
\textsc{Gin{\'e}, E. and R.~Nickl} (2016): \emph{Mathematical foundations of
  infinite-dimensional statistical models}, vol.~40, Cambridge University
  Press.

\bibitem[\protect\citeauthoryear{Gong, Peng, and Qi}{Gong
  et~al.}{2010}]{GongPengQi2010JMA}
\textsc{Gong, Y., L.~Peng, and Y.~Qi} (2010): \enquote{Smoothed jackknife
  empirical likelihood method for ROC curve,} \emph{Journal of Multivariate
  Analysis}, 101, 1520--1531.

\bibitem[\protect\citeauthoryear{Graham}{Graham}{2019}]{Graham2019Handbook}
\textsc{Graham, B.~S.} (2019): \enquote{Network data,} Tech. rep., National
  Bureau of Economic Research.

\bibitem[\protect\citeauthoryear{Graham, Niu, and Powell}{Graham
  et~al.}{2019}]{GrahamNiuPowell2019}
\textsc{Graham, B.~S., F.~Niu, and J.~L. Powell} (2019): \enquote{Kernel
  density estimation for undirected dyadic data,} \emph{arXiv preprint
  arXiv:1907.13630}.

\bibitem[\protect\citeauthoryear{Graham, Niu, and Powell}{Graham
  et~al.}{2021}]{graham2021minimax}
---\hspace{-.1pt}---\hspace{-.1pt}--- (2021): \enquote{Minimax risk and uniform
  convergence rates for nonparametric dyadic regression,} Tech. rep., National
  Bureau of Economic Research.

\bibitem[\protect\citeauthoryear{Hansen}{Hansen}{2008}]{Hansen2008ET}
\textsc{Hansen, B.~E.} (2008): \enquote{Uniform convergence rates for kernel
  estimation with dependent data,} \emph{Econometric Theory}, 24, 726--748.

\bibitem[\protect\citeauthoryear{Hinkley}{Hinkley}{1978}]{Hinkley1978}
\textsc{Hinkley, D.~V.} (1978): \enquote{Improving the jackknife with special
  reference to correlation estimation,} \emph{Biometrika}, 65, 13--21.

\bibitem[\protect\citeauthoryear{Hjort, McKeague, and Van~Keilegom}{Hjort
  et~al.}{2009}]{HjortMcKeagueVan_Keilegom2009AoS}
\textsc{Hjort, N.~L., I.~W. McKeague, and I.~Van~Keilegom} (2009):
  \enquote{Extending the scope of empirical likelihood,} \emph{Annals of
  Statistics}, 37, 1079--1111.

\bibitem[\protect\citeauthoryear{Jing, Yuan, and Zhou}{Jing
  et~al.}{2009}]{JingYuanZhang2009}
\textsc{Jing, B.-Y., J.~Yuan, and W.~Zhou} (2009): \enquote{Jackknife empirical
  likelihood,} \emph{Journal of the American Statistical Association}, 104,
  1224--1232.

\bibitem[\protect\citeauthoryear{Kallenberg}{Kallenberg}{2006}]{Kallenberg2006}
\textsc{Kallenberg, O.} (2006): \emph{Probabilistic Symmetries and Invariance
  Principles}, Springer Science \& Business Media.

\bibitem[\protect\citeauthoryear{Kato}{Kato}{2019}]{Kato2017lecture}
\textsc{Kato, K.} (2019): \enquote{Lecture notes on empirical process theory,}
  Tech. rep., Technical Report, Cornell University.

\bibitem[\protect\citeauthoryear{Kristensen}{Kristensen}{2009}]{Kristensen2009ET}
\textsc{Kristensen, D.} (2009): \enquote{Uniform convergence rates of kernel
  estimators with heterogeneous dependent data,} \emph{Econometric Theory}, 24,
  1433--1445.

\bibitem[\protect\citeauthoryear{Matsushita and Otsu}{Matsushita and
  Otsu}{2018}]{MatsushitaOtsu2018JER}
\textsc{Matsushita, Y. and T.~Otsu} (2018): \enquote{Likelihood inference on
  semiparametric models: Average derivative and treatment effect,}
  \emph{Japanese Economic Review}, 69, 133--155.

\bibitem[\protect\citeauthoryear{Matsushita and Otsu}{Matsushita and
  Otsu}{2021}]{MatsushitaOtsu2020}
---\hspace{-.1pt}---\hspace{-.1pt}--- (2021): \enquote{Jackknife empirical
  likelihood: small bandwidth, sparse network and high-dimension asymptotic,}
  \emph{Biometrika}, forthcoming.

\bibitem[\protect\citeauthoryear{McCullagh}{McCullagh}{2000}]{Mccullagh2000}
\textsc{McCullagh, P.} (2000): \enquote{Resampling and exchangeable arrays,}
  \emph{Bernoulli}, 6, 285--301.

\bibitem[\protect\citeauthoryear{Menzel}{Menzel}{2021}]{Menzel2020}
\textsc{Menzel, K.} (2021): \enquote{Bootstrap with cluster-dependence in two
  or more dimensions,} \emph{Econometrica}, 89, 2143--2188.

\bibitem[\protect\citeauthoryear{Owen}{Owen}{1990}]{Owen1990}
\textsc{Owen, A.} (1990): \enquote{Empirical likelihood ratio confidence
  regions,} \emph{Annals of Statistics}, 90--120.

\bibitem[\protect\citeauthoryear{Owen}{Owen}{1988}]{Owen1988}
\textsc{Owen, A.~B.} (1988): \enquote{Empirical likelihood ratio confidence
  intervals for a single functional,} \emph{Biometrika}, 75, 237--249.

\bibitem[\protect\citeauthoryear{Owen}{Owen}{2001}]{Owen200Book}
---\hspace{-.1pt}---\hspace{-.1pt}--- (2001): \emph{Empirical likelihood}, CRC
  press.

\bibitem[\protect\citeauthoryear{Owen}{Owen}{2007}]{Owen2007}
---\hspace{-.1pt}---\hspace{-.1pt}--- (2007): \enquote{The pigeonhole
  bootstrap,} \emph{Annals of Applied Statistics}, 1, 386--411.

\bibitem[\protect\citeauthoryear{Peng}{Peng}{2012}]{Peng2012CanJS}
\textsc{Peng, L.} (2012): \enquote{Approximate jackknife empirical likelihood
  method for estimating equations,} \emph{Canadian Journal of Statistics}, 40,
  110--123.

\bibitem[\protect\citeauthoryear{Peng, Qi, and Van~Keilegom}{Peng
  et~al.}{2012}]{PengOiVan_Keilegom2012}
\textsc{Peng, L., Y.~Qi, and I.~Van~Keilegom} (2012): \enquote{Jackknife
  empirical likelihood method for copulas,} \emph{Test}, 21, 74--92.

\bibitem[\protect\citeauthoryear{Schlenker and Walker}{Schlenker and
  Walker}{2016}]{SchlenkerWalker2016}
\textsc{Schlenker, W. and R.~W. Walker} (2016): \enquote{Airports, air
  pollution, and contemporaneous health,} \emph{Review of Economic Studies},
  83, 768--809.

\bibitem[\protect\citeauthoryear{Silverman}{Silverman}{1986}]{Silverman1986}
\textsc{Silverman, B.~W.} (1986): \emph{Density Estimation for Statistics and
  Data Analysis}, vol.~26, CRC press.

\bibitem[\protect\citeauthoryear{Snijders and Borgatti}{Snijders and
  Borgatti}{1999}]{SnijdersBorgatti1999}
\textsc{Snijders, T.~A. and S.~P. Borgatti} (1999): \enquote{Non-parametric
  standard errors and tests for network statistics,} \emph{Connections}, 22,
  161--170.

\bibitem[\protect\citeauthoryear{van~der Vaart and Wellner}{van~der Vaart and
  Wellner}{1996}]{vdVW1996}
\textsc{van~der Vaart, A.~W. and J.~A. Wellner} (1996): \emph{Weak Convergence
  and Empirical Processes}, Springer.

\bibitem[\protect\citeauthoryear{Wang, Peng, and Qi}{Wang
  et~al.}{2013}]{WangPengQi2013Sinica}
\textsc{Wang, R., L.~Peng, and Y.~Qi} (2013): \enquote{Jackknife empirical
  likelihood test for equality of two high dimensional means,} \emph{Statistica
  Sinica}, 667--690.

\bibitem[\protect\citeauthoryear{Zhang and Zhao}{Zhang and
  Zhao}{2013}]{ZhangZhao2013JMA}
\textsc{Zhang, Z. and Y.~Zhao} (2013): \enquote{Empirical likelihood for linear
  transformation models with interval-censored failure time data,}
  \emph{Journal of Multivariate Analysis}, 116, 398--409.

\end{thebibliography}

\end{document}